\documentclass[11pt,a4]{article}

\usepackage{amssymb,amsmath}
\usepackage{pstricks, pst-coil, pst-node, pst-tree, multido}
\usepackage[english]{babel}
\usepackage{graphics}

\newtheorem{DE}{Definition}[section]

\newcommand {\sm} {\setminus}

\usepackage{latexsym} 
\usepackage{theorem} 
\newcommand{\qed}{\relax\ifmmode\hskip2em\Box\else\unskip\nobreak\hfill$\Box$\fi}

\newcommand{\tp}{\!-\!}

\newtheorem{theorem}[DE]{Theorem}
\newtheorem{lemma}[DE]{Lemma}

{\theoremstyle{break}\theorembodyfont{\rmfamily}}
{\theoremstyle{break}\theorembodyfont{\rmfamily}}

\newcounter{claim}
\newenvironment{proof}[1][]%
	{\noindent {\setcounter{claim}{0}\sc proof --- }{#1}{}}{\qed\vspace{2ex}}
\newenvironment{claim}[1][]%
	{\refstepcounter{claim}\vspace{1ex}\noindent {(\it\arabic{claim}) {#1}{}}\it}{\vspace{1ex}}
\newenvironment{proofclaim}[1][]%
	{\noindent {}{#1}{}}{ This proves~(\arabic{claim}).\vspace{1ex}}

\newcommand{\C}[2]{{\cal C}^{\text{\scriptsize\sc #1}}_{\text{\scriptsize\sc #2}}}

\newcommand{\R}{\ensuremath{\mathbb{R}}}

\newcommand{\N}{\ensuremath{\mathbb{N}}}

\begin{document}

\title{Combinatorial optimization with 2-joins}

\date{July 4, 2011
}
\author{Nicolas Trotignon\thanks{CNRS, LIP -- ENS Lyon (France),\
    email: nicolas.trotignon@ens-lyon.fr.  Partially supported by the
    French \emph{Agence Nationale de la Recherche} under reference
    \textsc{anr Heredia 10 jcjc 0204 01}.}~~and Kristina
  Vu\v{s}kovi\'c\thanks{School of Computing, University of Leeds,
    Leeds LS2 9JT, UK and Faculty of Computer Science, Union
    University, Knez Mihailova 6/VI, 11000 Belgrade, Serbia. E-mail:
    k.vuskovic@leeds.ac.uk.  Partially supported by Serbian Ministry
    of Education and Science grants III44006 and OI174033 and EPSRC
    grant EP/H021426/1.  \newline The two authors are also supported
    by PHC Pavle Savi\'c grant jointly awarded by EGIDE, an agency of
    the French Minist\`ere des Affaires \'etrang\`eres et
    europ\'eennes, and Serbian Ministry for Science and Technological
    Development. }}

\maketitle

\begin{abstract}
  A 2-join is an edge cutset that naturally appears in decomposition of
  several classes of graphs closed under taking induced subgraphs,
  such as perfect graphs and claw-free graphs.  In this paper we
  construct combinatorial polynomial time algorithms for finding a
  maximum weighted clique, a maximum weighted stable set and an
  optimal coloring for a class of perfect graphs decomposable by
  2-joins: the class of perfect graphs that do not have a balanced
  skew partition, a 2-join in the complement, nor a homogeneous
  pair. The techniques we develop are general enough to be easily
  applied to finding a maximum weighted stable set for another class
  of graphs known to be decomposable by 2-joins, namely the class of
  even-hole-free graphs that do not have a star cutset.

  We also give a simple class of graphs decomposable by 2-joins into
  bipartite graphs and line graphs, and for which finding a maximum
  stable set is NP-hard.  This shows that having holes all of the same
  parity gives essential properties for the use of 2-joins in
  computing stable sets.
\end{abstract}



\maketitle

\section{Introduction}

In this paper all graphs are simple and finite.  We say that a graph
$G$ {\em contains} a graph $F$ if $F$ is isomorphic to an induced
subgraph of $G$, and it is $F$-free if it does not contain $F$.  A
\emph{hole} in a graph is an induced cycle of length at least~4.  An
\emph{antihole} is the complement of a hole.  A graph $G$ is said to
be \emph{perfect} if for every induced subgraph $G'$ of $G$, the
chromatic number of $G'$ is equal to the maximum size of a clique
of~$G'$.  A graph is said to be \emph{Berge} if it does not contain an
odd hole nor an odd antihole.  In 1961, Berge~\cite{berge:61}
conjectured that every Berge graph is perfect.  This was known as the
\emph{Strong Perfect Graph Conjecture} (SPGC), it was an object of
much research until it was finally proved by Chudnovsky, Robertson,
Seymour and Thomas in 2002~\cite{chudnovsky.r.s.t:spgt}.  So Berge
graphs and perfect graphs are the same class of graphs, but we prefer
to write ``Berge'' for results which rely on the structure of the
graphs, and ``perfect'' for results which rely on the properties of
their colorings.  We now explain the motivation for this paper and
describe informally the results.  We use several technical notions
that will be defined precisely later.

\subsection{Optimization with decomposition}

In the 1980's, Gr\"ostchel, Lov\'asz and Schrijver~\cite{gls,
  gls:color} devised a polynomial time algorithm that optimally colors
any perfect graph.  This algorithm relies on the ellipsoid method and
consequently is impractical.  Finding a purely combinatorial
polynomial time algorithm is still an open question.  In fact, after
the resolution of the SPGC and the construction of polynomial time
recognition algorithm for Berge graphs \cite{cclsv}, this is the key
open problem in the area.

The proof of the SPGC in~\cite{chudnovsky.r.s.t:spgt} was obtained
through a \emph{decomposition theorem} for Berge graphs.  So, it is a
natural question to ask whether this decomposition theorem can be used
for coloring and other combinatorial optimization problems.  Up to
now, it seems that the decomposition theorem is very difficult to use.
Let us explain why.  In a connected graph $G$, a subset of vertices
and edges is a {\em cutset} if its removal disconnects $G$.  A
Decomposition Theorem for a class of graphs ${\cal C}$ is of the
following form.

\vspace{2ex}

\noindent
{\bf Decomposition Theorem:} {\em If $G$ belongs to ${\cal C}$
then $G$ is either ``basic'' or $G$ has some particular cutset.}

\vspace{2ex}

Decomposition Theorems can be used for proving theorems. For example,
the SPGC was proved using the decomposition theorem for Berge
graphs~\cite{chudnovsky.r.s.t:spgt}, by ensuring that ``basic'' graphs
are simple in the sense that they are easily proved to be perfect
directly, and the cutsets used have the property that they
cannot occur in a minimum counter-example to the SPGC.

Decomposition theorems can be used also for algorithms.  For instance,
they yielded many recognition algorithms.  To recognize a class ${\cal
  C}$ with a decomposition theorem, ``basic'' graphs need to be simple
in the sense that they can easily be recognized, and the cutsets used
need to have the following property.  The removal of a cutset  from
a graph $G$ disconnects $G$ into two or more connected components.
From these components {\em blocks of decomposition} are constructed by
adding some more vertices and edges. A decomposition is {\em ${\cal
    C}$-preserving} if it satisfies the following: $G$~belongs to
${\cal C}$ if and only if all the blocks of decomposition belong to
${\cal C}$. A recognition algorithm takes a graph $G$ as input and
decomposes it using ${\cal C}$-preserving decompositions into a
polynomial number of basic blocks, which are then checked, in
polynomial time, whether they belong to ${\cal C}$.  This is an ideal
scenario, and it worked for example for obtaining recognition
algorithms for regular matroids (using $k$-separations, $k=1,2,3$)
\cite{sey}, max-flow min-cut matroids (using 2-sums and $\Delta$-sums)
\cite{tr1}, and graphs that do not contain a cycle with a unique chord
(using 1-joins and vertex cutsets of size 1 or 2) \cite{tv}.

But several classes of graphs are too complex for allowing such a
direct approach.  The main problem is what we call \emph{strong
  cutsets}.  The typical example of a strong cutset is the
Chv{\'a}tal's \emph{star cutset}~\cite{chvatal:starcutset}: a cutset
that contains one vertex and a subset of its neighbors.  The problem
with such a cutset is that it can be very big, for instance, it can
be the whole vertex-set except two vertices.  And since in the cutset
itself, edges are quite unconstrained, knowing that the graph has a
star cutset tells little about its structure.  From this discussion,
it could even be thought that star cutsets are just useless, but this
is not the case: deep theorems use strong cutsets.  The first one is
the Hayward's decomposition theorem of weakly triangulated
graphs~\cite{hayward:wt}, a simple class of graphs that captures
ideas that were used later for all Berge graphs.

More generally, using strong cutsets is essential for proving theorems
about many complex classes of graphs closed under taking induced
subgraphs, the most famous example being the proof of the SPGC that
uses a generalization of star cutsets: the \emph{balanced skew
  partition}.  Robertson and Seymour have obtained
results about minor-closed families of stunning generality, see
\cite{lo2} for a survey.  The Robertson-Seymour Theorem \cite{rs2}
states that every minor-closed class of graphs can be characterized by
a finite family of excluded minors.  Furthermore, every minor-closed
property of graphs can be tested in polynomial time \cite{rs1}.  The
fact that a unified theory with deep algorithmic consequences exists
for classes closed under taking minors and that no such theory
exists up to now for the induced subgraph containment relation has
perhaps something to do with these strong cutsets.

Yet, for recognition algorithms, strong cutsets can be used.  Examples
are balanced matrices (using 2-joins and double star cutsets)
\cite{ccr}, balanced $0, \pm 1$ matrices (using 2-joins, 6-joins and
double star cutsets) \cite{cckv-bal}, even-hole-free graphs (using
2-joins and star cutsets) \cite{dsv-ehf}, and Berge graphs (using
2-joins and double star cutsets from the decomposition theorem in
\cite{ccv-ohf}) \cite{cclsv}.  This is accomplished by a powerful
tool: the \emph{cleaning}, that is a preprocessing of graphs not worth
describing here.  For combinatorial optimization algorithms (maximum
clique, coloring, \dots), it seems that the cleaning is useless and no
one knows how strong cutsets could be used.

\subsection{Our results}

What we are interested in is whether the known decomposition theorems
for perfect graphs \cite{chudnovsky.r.s.t:spgt, chudnovsky:trigraphs,
  nicolas:bsp} and even-hole-free graphs \cite{cckv-ehf1,dsv-ehf} can
be used to construct combinatorial polynomial time optimization
algorithms.  But as we explained above, we do not know how to handle
the strong cutsets (namely star cutsets and their generalizations,
balanced skew partitions and double star cutsets).  So we take the
\emph{bottom-up approach}.  Let us explain this.  In all classes
similar to Berge graphs (in the sense that strong cutsets are needed
for their decomposition), it can be proved that a decomposition tree
can be built by using in a first step only the strong cutsets, and in
a second step only the other cutsets (this is not at all obvious for
Berge graphs, see \cite{nicolas:bsp}).  So it is natural to ask
whether we can optimize on classes of graphs decomposable by cutsets
that are not strong.

For Berge graphs and even-hole-free graphs, if we assume that no
strong cutset is needed, we obtain a class of graphs decomposable
along \emph{2-joins}, a decomposition that was introduced by
Cornu\'ejols and Cunningham in \cite{cornuejols.cunningham:2join}
where they prove that no minimum counter-example to the SPGC can admit
a 2-join.  2-Joins proved to be of great use in decomposition
theorems, they were also used in several recognition algorithms
mentioned above, but never yet have they been used in construction of
optimization algorithms.  Proving that a minimally imperfect graph
admits no 2-join is done by building blocks of decomposition w.r.t.\ a
2-join that are smaller graphs with the same clique number and the
same chromatic number as the original graph.  But as we will see, it
is not at all straightforward to transform these ideas into
optimization algorithms for our classes.

Our main results are Theorem~\ref{l:CS} and~\ref{th:color}.  They say
that for Berge graphs with no balanced skew partition, no 2-join in
the complement and no homogeneous pair, the following problems can be
solved combinatorially in polynomial time: maximum weighted clique,
maximum weighted stable set and optimal coloring.  The homogeneous
pair and the 2-join in the complement are not really strong cutsets.
Excluding homogeneous pairs was suggested to us by Celina de
Figueiredo~\cite{celina:PC} and is very helpful for several technical
reasons, see below.  In this bottom-up approach, the next step would
be to analyze how homogeneous pairs could be used in optimization
algorithms.  This step might be doable because some classes of Berge
graphs are optimized with homogeneous pairs,
see~\cite{figuereido.m:bullfreeopt}.  This might finally lead to a
coloring algorithm for Berge graphs with no balanced skew partitions.

Our approach is general enough to give results about even-hole-free
graphs that are structurally quite similar to Berge graphs. Their
structure was first studied by Conforti, Cornu\'ejols, Kapoor and
Vu\v{s}kovi\'c in \cite{cckv-ehf1} and \cite{cckv-ehf2}. They were
focused on showing that even-hole-free graphs can be recognized in
polynomial time (a problem that at that time was not even known to be
in NP), and their primary motivation was to develop techniques which
can then be used in the study of Berge graphs. In \cite{cckv-ehf1} a
decomposition theorem was obtained using 2-joins, star, double star
and triple star cutsets, and in \cite{cckv-ehf2} a polynomial time
decomposition based recognition algorithm was constructed. Later da
Silva and Vu\v{s}kovi\'c \cite{dsv-ehf} significantly strengthened the
decomposition theorem for even-hole-free graphs by using just 2-joins
and star cutsets, which significantly improved the running time of the
recognition algorithm for even-hole-free graphs.  It is this
strengthening that we use in this paper.  One can find a maximum
clique in an even-hole-free graph in polynomial time, since as
observed by Farber \cite{farber} 4-hole-free graphs have ${\cal O}
(n^2)$ maximal cliques and hence one can list them all in polynomial
time (in all complexity analysis, $n$ stands for the number of
vertices of the input graph and $m$ for the number of its edges).  In
\cite{dsv} da Silva and Vu\v{s}kovi\'c show that every even-hole-free
graph has a vertex whose neighborhood is hole-free, which leads to a
faster algorithm for finding a maximum clique in an even-hole-free
graph. The complexities of finding a maximum stable set and an optimal
coloring are not known for even-hole-free graphs.

\subsection{Outline of the paper}

In Section~\ref{sec:decomp} we precisely describe all the
decomposition theorems we will be working with.  For even-hole-graphs,
we rely on the theorem of da Silva and Vu\v skovi\'c~\cite{dsv}.  For
the decomposition of Berge graphs we rely on an improvement due to
Trotignon~\cite{nicolas:bsp} of the decomposition theorems of
Chudnovsky, Robertson, Seymour and Thomas
\cite{chudnovsky.r.s.t:spgt}, and Chudnovsky
\cite{chudnovsky:trigraphs}.  We need this improvement because we use
the so called non-path 2-joins in the algorithms, and not simply the
2-joins as defined in \cite{chudnovsky.r.s.t:spgt}.  For the same
reason, we need to exclude the homogeneous pair because some Berge
graphs are decomposable only along path 2-join or homogeneous pair (an
example is represented Figure~\ref{fig:contrex3}).

In Section~\ref{sec:blocks} we show how to construct blocks of
decomposition w.r.t.\ 2-joins that will be class-preserving. This
allows us to recursively decompose along 2-joins down to basic graphs.

Using 2-joins in combinatorial optimization algorithms requires
building blocks of decomposition and asking at least two questions for
at least one block (while for recognition algorithms, one question is
enough).  When this process is recursively applied it can potentially
lead to an exponential blow-up even when the decomposition tree is
linear in the size of the input graph.  This problem is bypassed by
using what we call \emph{extreme 2-joins}, that is 2-joins whose one
block of decomposition is basic.  In Section~\ref{sec:extrem} we prove
that non-basic graphs in our classes actually have extreme 2-joins.
Interestingly, we give an example showing that Berge graphs in general
do not necessarily have extreme 2-joins, their existence is a special
property of graphs with no star cutset.  This allows us to build a
decomposition tree in which every internal node has two children, at
least one of which is a leaf, and hence corresponds to a basic graph.

In Section~\ref{sec:clique}, we show how to put weights on vertices of
the block of decomposition w.r.t.\ an extreme 2-join in order to
compute maximum cliques.  In fact the approach used here could solve
the maximum weighted clique problem for any class with a decomposition
theorem along extreme 2-joins down to basic graphs for which the
problem can be solved.

For stable sets, the problem is more complicated.  As an evidence, in
Section~\ref{sec:npc}, we show a simple class of graphs decomposable
along extreme 2-joins into bipartite graphs and line graphs of cycles
with one chord.  This class has a structure close to Berge graphs and
in fact much simpler in many respects.  Yet, we prove that computing
maximum stable sets for this class is NP-hard.  So, in
Section~\ref{sec:alphaTrack}, devoted to stable sets, we need to
somehow take advantage of the parity of the cycles.  To do so, in
Subsection~\ref{sec:ineq}, we prove a couple of lemmas showing that a
maximum weighted stable set and a 2-join overlap in a very special way
for graphs where cycles are all of the same parity.  These lemmas
allow an unusual construction for blocks that preserve simultaneously
the weight of a maximal weighted stable set and being Berge.

Our unusual blocks raise some problems.  First, if we use them to
fully decompose a graph from our class, what we obtain in the leaves
of the decomposition tree are not basic graphs, but what we call 
\emph{extensions} of  basic graphs.  In Section~\ref{sec:extension},
we show how to solve optimization problems for extensions of basic
graphs.

Another problem (that is in fact the source of the previous one) is
that our blocks are not class-preserving.  They do preserve being
Berge, but they introduce balanced skew partitions.  To bypass this
problem, we construct our decomposition tree in Section~\ref{sec:tree}
in two steps.  First, we use classical class-preserving blocks.  In
the second step, we reprocess the tree to use the unusual blocks.   

In Section~\ref{sec:color} we give the algorithms for solving the
clique and stable set problems.  We also recall a classical method to
color a perfect graph assuming that subroutines exist for cliques and
stable sets.  We show that this method can be used for our class.

Section~\ref{sec:npc} is devoted to the NP-hardness result mentioned
above.

\section{Decomposition theorems} 
\label{sec:decomp}

In this section we introduce all the decomposition theorems
we will use in this paper.
But before we continue, for the convenience we first establish the following 
notation for the classes of graphs we will be working with.


We denote by $\cal C$ the class of all graphs.  We use the superscript
{\sc parity} to mean that all holes have the same parity.  So,
$\C{parity}{}$ can be defined equivalently as the union of the
odd-hole-free graphs and the even-hole-free graphs.  Note that every
Berge graph is in $\C{parity}{}$.  We will use the superscript {\sc
  ehf} to restrict the class to even-hole-free graphs and {\sc Berge}
to restrict the class to Berge graphs.  So for instance, $\C{Berge}{}$
denotes the class of Berge graphs.  We use the subscript {\sc no
  cutset} to restrict the class to those graphs that do not have a
balanced skew partition, a connected non-path 2-join in the
complement, nor a homogeneous pair.  For technical reasons, mainly to
avoid reproving results from~\cite{nicolas:bsp}, we also need the
subscript {\sc no bsp} to restrict a class to graphs with no balanced
skew partition.  We use the subscript {\sc no sc} to restrict the
class to graphs with no star cutset.  We use the subscript {\sc basic}
to restrict the class to the relevant basic graphs.
Table~\ref{tab:class} sums up all the classes used in this paper.  The
classes are defined more formally in the remainder of this section.

\begin{table}[h]
\begin{center}
\begin{tabular}{ll}
Class          & Definition                                 \\\hline
\rule{0em}{3ex}$\C{parity}{}$         
& {\parbox[t]{9cm}{Graphs where all holes have same parity}}\\
\rule{0em}{3ex}$\C{Berge}{}$          
& {\parbox[t]{9cm}{Berge graphs}}                           \\
\rule{0em}{3ex}$\C{ehf}{}$            
& {\parbox[t]{9cm}{Graphs that do not contain even holes}}              \\
\rule{0em}{3ex}$\C{Berge}{no cutset}$ 
& {\parbox[t]{9cm}{Berge graphs with no balanced skew                                         
partition, no connected non-path 2-join in the complement and no homogeneous pair}}            \\
\rule{0em}{3ex}$\C{Berge}{no bsp}$ 
& {\parbox[t]{9cm}{Berge graphs with no balanced skew
                                              partition}}            \\
\rule{0em}{3ex}$\C{Berge}{basic}$     
& {\parbox[t]{9cm}{Bipartite,  
    line graphs of bipartite, path-cobipartite and
           path-double split graphs;
          complements of all these graphs}}           \\
\rule{0em}{3ex}$\C{ehf}{basic}$      
 & {\parbox[t]{9cm}{Even-hole-free graphs that can be 
obtained from the line graph of a tree by adding at most two vertices}}        \\
\rule{0em}{3ex}$\C{}{no sc}$          
& {\parbox[t]{9cm}{Graphs that have no star cutset}}
\\
\rule{0em}{3ex}$\C{parity}{no sc}$          
& {\parbox[t]{9cm}{Graphs of $\C{parity}{}$ that have no star cutset}}
\\
\rule{0em}{3ex}$\C{ehf}{no sc}$          
& {\parbox[t]{9cm}{Even-hole-free graphs that have no star cutset}}
\end{tabular}
\caption{Classes of graphs\label{tab:class}}
\end{center}
\end{table}

We call \emph{path} any connected graph with at least one vertex of
degree at most~1 and no vertex of degree greater than~2. A path has at
most two vertices of degree~1, which are the \emph{ends} of the
path. If $a, b$ are the ends of a path $P$ we say that $P$ is
\emph{from $a$ to~$b$}. The other vertices are the \emph{interior}
vertices of the path. We denote by $v_1 \tp \cdots \tp v_n$ the path
whose edge set is $\{v_1v_2, \dots, v_{n-1}v_n\}$.  When $P$ is a
path, we say that $P$ is \emph{a path of $G$} if $P$ is an induced
subgraph of $G$. If $P$ is a path and if $a, b$ are two vertices of
$P$ then we denote by $a \tp P \tp b$ the only induced subgraph of $P$
that is path from $a$ to $b$.  The \emph{length} of a path is the
number of its edges. An \emph{antipath} is the complement of a path.
Let $G$ be a graph and let $A$ and $B$ be two subsets of $V(G)$. A
path of $G$ is said to be \emph{outgoing from $A$ to $B$} if it has an
end in $A$, an end in $B$, length at least~2, and no interior vertex
in $A\cup B$.

The 2-join was first defined by Cornu\'ejols and
Cunningham~\cite{cornuejols.cunningham:2join}.  A partition $(X_1,
X_2)$ of the vertex-set is a \emph{2-join} if for $i=1,2$, there exist
disjoint non-empty $A_i, B_i \subseteq X_i$ satisfying the following:

\begin{itemize} 
\item
  every vertex of $A_1$ is adjacent to every vertex of $A_2$ and every
  vertex of $B_1$ is adjacent to every vertex of $B_2$;
\item
  there are no other edges between $X_1$ and $X_2$;
\item 
  for $i=1,2$, $|X_i| \geq 3$; 
\item 
  for $i=1,2$, $X_i$ is not a path of length~2 with an end in $A_i$,
  an end in $B_i$ and its unique interior vertex in $C_i = X_i\sm (A_i
  \cup B_i)$.
\end{itemize}

The sets $X_1, X_2$ are the two \emph{sides} of the 2-join.  When sets
$A_i$'s and $B_i$'s are like in the definition we say that $(X_1, X_2,
A_1, B_1, A_2, B_2)$ is a \emph{split} of $(X_1, X_2)$.  Implicitly,
for $i= 1, 2$, we will denote by $C_i$ the set $X_i \setminus (A_i
\cup B_i)$.

A 2-join $(X_1, X_2)$ in a graph $G$ is said to be \emph{connected} if
for $i= 1, 2$, there exists a path from $A_i$ to $B_i$ with interior in $C_i$.

A 2-join is said to be a \emph{path 2-join} if it has a split $(X_1,
X_2, A_1, B_1, A_2, B_2)$ such that for some $i\in \{1, 2\}$, $G[X_i]$
is a path with an end in $A_i$, an end in $B_i$ and interior in $C_i$.
Implicitly we will then denote by $a_i$ the unique vertex in $A_i$ and
by $b_i$ the unique vertex in $B_i$.  We say that $X_i$ is the
\emph{path-side} of the 2-join.  Note that when $G$ is not a hole then
at most one of $X_1, X_2$ is a path side of $(X_1, X_2)$. A
\emph{non-path 2-join} is a 2-join that is not a path 2-join.
Note that all the 2-joins used in \cite{cckv-bal}, \cite{cckv-ehf1}, 
\cite{cckv-ehf2}, \cite{ccr} \cite{conforti.c.v:square} and \cite{ccv-ohf}
are in fact non-path 2-joins.

\subsection{Decomposition of even-hole-free graphs}

A \emph{vertex cutset} in a graph $G$ is a set $S\subset V(G)$ such
that $G\setminus S$ is disconnected ($G\setminus S$ means $G[V(G)
\setminus S]$).  By $N[x]$ we mean $N(x) \cup \{x\}$.  A {\em star
  cutset} in a graph $G$ is a vertex cutset $S$ such that for some $x
\in S$, $S \subseteq N[x]$. Such a vertex $x$ is called a {\em center}
of the star, and we say that $S$ is {\em centered} at $x$.

A graph is in $\C{ehf}{basic}$ if it is even-hole-free and one can
obtain the line graph of a tree by deleting at most two of its
vertices.

Building on the work in \cite{kmv}, da Silva and Vu\v{s}kovi\'c
establish the following strengthening of the original decomposition
theorem for even-hole-free graphs \cite{cckv-ehf1}.

\begin{theorem}
[da Silva and Vu\v{s}kovi\'c \cite{dsv-ehf}]
\label{ehf}
If $G\in \C{ehf}{}$ then either $G\in \C{ehf}{basic}$ or $G$ has a
star cutset or a connected non-path 2-join.
\end{theorem}

Actually in the decomposition theorem of~\cite{dsv-ehf}, the basic graphs
are defined in a more specific way, but for the purposes of the algorithms
the statement of Theorem~\ref{ehf} suffices.

\subsection{Decomposition of Berge graphs}

If $X, Y \subseteq V(G)$ are disjoint, we say that $X$ is
\emph{complete} to $Y$ if every vertex in $X$ is adjacent to every
vertex in $Y$. We also say that $(X, Y)$ is a \emph{complete pair}. We
say that $X$ is \emph{anticomplete} to $Y$ if there are no edges
between $X$ and $Y$. We also say that $(X, Y)$ is an
\emph{anticomplete pair}.  We say that a graph $G$ is
\emph{anticonnected} if its complement $\overline{G}$ is connected.

Skew partitions were first introduced by
Chv\'atal~\cite{chvatal:starcutset}. A \emph{skew partition} of a
graph $G = (V,E)$ is a partition of $V$ into two sets $A$ and $B$ such
that $A$ induces a graph that is not connected, and $B$ induces a
graph that is not anticonnected. When $A_1, A_2, B_1, B_2$ are
non-empty sets such that $(A_1, A_2)$ partitions $A$, $(A_1, A_2)$ is
an anticomplete pair, $(B_1, B_2)$ partitions $B$, and ($B_1, B_2$) is
a complete pair, we say that $(A_1, A_2, B_1, B_2)$ is a \emph{split}
of the skew partition $(A, B)$. A \emph{balanced skew partition}
(first defined in~\cite{chudnovsky.r.s.t:spgt}) is a skew partition
$(A, B)$ with the additional property that every induced path of
length at least~2 with ends in $B$, interior in $A$ has even length,
and every antipath of length at least~2 with ends in $A$, interior in
$B$ has even length. If $(A, B)$ is a skew partition, we say that $B$
is a \emph{skew cutset}. If $(A, B)$ is balanced we say that the skew
cutset $B$ is \emph{balanced}. Note that Chudnovsky et
al.~\cite{chudnovsky.r.s.t:spgt} proved that no minimum
counter-example to the strong perfect graph conjecture admits
a balanced skew partition.

Call \emph{double split graph} (first defined
in~\cite{chudnovsky.r.s.t:spgt}) any graph $G$ that may be constructed
as follows.  Let $k,l \geq 2$ be integers. Let $A = \{a_1, \dots,
a_k\}$, $B= \{b_1, \dots, b_k\}$, $C= \{c_1, \dots, c_l\}$, $D= \{d_1,
\dots, d_l\}$ be four disjoint sets. Let $G$ have vertex-set $A\cup B
\cup C \cup D$ and edges in such a way that:

\begin{itemize}
\item 
  $a_i$ is adjacent to $b_i$ for $1 \leq i \leq k$.  There are no
  edges between $\{a_i, b_i\}$ and $\{a_{i'}, b_{i'}\}$ for $1\leq i <
  i' \leq k$;
\item 
  $c_j$ is non-adjacent to $d_j$ for $1 \leq j \leq l$. There are all
  four edges between $\{c_j, d_j\}$ and $\{c_{j'}, d_{j'}\}$ for
  $1\leq j < j' \leq l$;
\item
  there are exactly two edges between $\{a_i, b_i\}$ and $\{c_j,
  d_j\}$ for $1\leq i \leq k$, $1 \leq j \leq l$ and these two
  edges are disjoint.
\end{itemize}

The homogeneous pair was first defined by Chv\'atal and
Sbihi~\cite{chvatal.sbihi:bullfree}.  The definition that we give here
is a slight variation.  A \emph{homogeneous pair} is a partition of
$V(G)$ into six sets $(A, B, C, D, E, F)$ such that:

\begin{itemize} 
\item $A$, $B$, $C$, $D$ and $F$ are non-empty (but $E$ is possibly empty);
\item 
  every vertex in $A$ has a neighbor in $B$ and a non-neighbor in $B$,
  and vice versa (note that this implies that $A$ and $B$ both contain
  at least 2 vertices); 
\item the pairs $(C,A)$, $(A,F)$, $(F,B)$, $(B,D)$ are complete; 
\item the pairs $(D,A)$, $(A,E)$, $(E,B)$, $(B,C)$ are anticomplete. 
\end{itemize}

All the decomposition theorems for Berge graphs that we mention now
are published in papers that have a definition of a connected 2-join
and a homogeneous pair slightly more restrictive than ours.  So, the
statements that we give here follow directly from the original
statements.

The following theorem was first conjectured in a slightly different
form by Conforti, Cornu\'ejols and Vu\v skovi\'c, who proved it in the
particular case of square-free graphs~\cite{conforti.c.v:square}.  A
corollary of it is the Strong Perfect Graph Theorem.

\begin{theorem}[Chudnovsky, Robertson, Seymour and Thomas, \cite{chudnovsky.r.s.t:spgt}]
  \label{th.0}
  Let $G$ be a Berge graph. Then either $G$ is bipartite, line graph
  of bipartite, complement of bipartite, complement of line graph of
  bipartite or double split, or $G$ has a homogeneous pair, or $G$ has
  a balanced skew partition or one of $G, \overline{G}$ has a
  connected 2-join.
\end{theorem} 

The theorem that we state now is due to Chudnovsky who proved it from
scratch, that is without assuming Theorem~\ref{th.0}.  Her proof uses
the notion of \emph{trigraph}.  The theorem shows that homogeneous
pairs are not necessary to decompose Berge graphs.  Thus it is a
result stronger than Theorem~\ref{th.0}.

\begin{theorem}[Chudnovsky, \cite{chudnovsky:trigraphs,chudnovsky:these}]
\label{th.1}
Let $G$ be a Berge graph. Then either $G$ is bipartite, line graph of
bipartite, complement of bipartite, complement of line graph of
bipartite or double split, or one of $G, \overline{G}$ has a
connected 2-join or $G$ has a balanced skew partition.
\end{theorem}

\subsection{Avoiding path 2-joins in Berge graphs}

Theorem~\ref{th.1} allows path 2-joins and they are not easy to
handle.  Because their path side is sometimes not substantial enough
to allow building a block of decomposition that carries sufficiently
information.  But there are other reasons that we explain now.  Let us
first remind the starting point of this work: we do not know how to
handle skew partitions in algorithms.  So, things should be easier for
a Berge graph with no skew partition.  Such a graph is likely to have
a 2-join, but when decomposing along this 2-join, we may create a skew
partition again.  Thus, it seems impossible to devise a recursive
algorithm that decomposes graphs with no balanced skew partitions
along 2-joins.  A careful study of this phenomenon, done in
\cite{nicolas:bsp}, shows that path 2-joins, more precisely certain
kinds of path 2-join, are responsible for this and can be avoided.
Let us state this more precisely.

The following theorem shows that path 2-joins are not necessary to
decompose Berge graphs, but at the expense of extending balanced skew
partitions to general skew partitions and introducing a new basic
class. So, this theorem is useless for us (at least, we do not know
how to use it).  Before stating the theorem, we need to define the new
basic class.

A graph $G$ is path-cobipartite if it is a Berge graph obtained by
subdividing an edge between the two cliques that partitions the
complement of a bipartite graph.  More precisely, a graph is
\emph{path-cobipartite} if its vertex-set can be partitioned into
three sets $A, B, P$ where $A$ and $B$ are non-empty cliques and $P$
consist of vertices of degree~2, each of which belongs to the interior
of a unique path of odd length with one end $a$ in $A$, the other one
$b$ in $B$. Moreover, $a$ has neighbors only in $A \cup P$ and $b$ has
neighbors only in $B \cup P$. Note that a path-cobipartite graph such
that $P$ is empty is the complement of bipartite graph. Note that our
path-cobipartite graphs are simply the complement of the
\emph{path-bipartite} graphs defined by Chudnovsky
in~\cite{chudnovsky:these}.  For convenience, we prefer to think about
them in the complement as we do.

\begin{theorem}[Chudnovsky, \cite{chudnovsky:these}]
\label{th.2}
Let $G$ be a Berge graph. Then either $G$ is bipartite, line graph of
bipartite, complement of bipartite, complement of line graph of
bipartite, double split, path-bipartite, complement of path-bipartite,
or $G$ has a connected non-path 2-join, or $\overline{G}$ has a
connected 2-join, or $G$ has a homogeneous pair or $G$ has a skew
partition.
\end{theorem}

A \emph{path-double split graph} is any graph $H$ that may be
constructed as follows.  Let $k,l \geq 2$ be integers. Let $A = \{a_1,
\dots, a_k\}$, $B= \{b_1, \dots, b_k\}$, $C= \{c_1, \dots, c_l\}$, $D=
\{d_1, \dots, d_l\}$ be four disjoint sets. Let $E$ be another
possibly empty set disjoint from $A$, $B$, $C$, $D$. Let $H$ have
vertex-set $A\cup B \cup C \cup D \cup E$ and edges in such a way
that:

\begin{itemize}
\item for every  vertex $v$ in $E$, $v$ has degree~2 and there exists
  $i \in \{1, \dots, k\}$ such that $v$ lies on a
  path of odd length from $a_i$ to $b_i$; 
\item 
  for $1 \leq i \leq k$, there is a unique path of odd length
  (possibly~1) between $a_i$ and $b_i$ whose interior is in $E$.
  There are no edges between $\{a_i, b_i\}$ and $\{a_{i'}, b_{i'}\}$
  for $1\leq i < i' \leq k$;
\item 
  $c_j$ is non-adjacent to $d_j$ for $1 \leq j \leq l$. There are all
  four edges between $\{c_j, d_j\}$ and $\{c_{j'}, d_{j'}\}$ for
  $1\leq j < j' \leq l$;
\item
  there are exactly two edges between $\{a_i, b_i\}$ and $\{c_j,
  d_j\}$ for $1\leq i \leq k$, $1 \leq j \leq l$ and these two
  edges are disjoint.
\end{itemize}

Note that a path-double split graph $G$ has an obvious skew partition
that is not balanced: $(A\cup B\cup E, C\cup D)$.  In fact, it is
proved in~\cite{nicolas:bsp}, Lemma~4.5, that this is the unique skew
partition of $G$.  Also, either $E$ is empty and the graph is a double
split graph or $E$ is not empty and the graph has a path 2-join.
Path-double split graphs are the reason why in Theorem~\ref{th.2}, one
needs to add non-balanced skew partitions in the list of
decompositions.

We call \emph{flat path of a graph $G$} any path of length at least~2,
whose interior vertices all have degree~2 in $G$ and whose ends have
no common neighbors outside the path.  A \emph{homogeneous 2-join} is
a partition of $V(G)$ into six non-empty sets $(A,$ $B,$ $C,$ $D,$
$E,$ $F)$ such that:

\begin{itemize} 
\item 
  $(A, B, C, D, E, F)$ is a homogeneous pair such that $E$ is not empty;
\item 
  every vertex in $E$ has degree~2 and belongs to a flat path of odd
  length with an end in $C$, an end in $D$ and whose interior is in
  $E$;
 \item 
   every flat path outgoing from $C$ to $D$ and whose interior is in
   $E$ is the path-side of a non-cutting connected 2-join of $G$.
\end{itemize} 

Note we have not defined \emph{cutting} and \emph{non-cutting}
2-joins.  The definition is long (see~\cite{nicolas:bsp}) and we do
not need it here because the only property of homogeneous 2-joins that we
are going to use is that they imply the existence of a homogeneous
pair.  Homogeneous 2-joins are the reason why in Theorem~\ref{th.2},
one needs to add homogeneous pairs in the list of decompositions.

The following theorem generalizes the previously known decomposition
theorems for Berge graphs.  So it implies the Strong Perfect Graph
Theorem, but its proof relies heavily on Theorem~\ref{th.1}.  Hence it
does not give a new proof of the Strong Perfect Graph Theorem.

\begin{theorem}[Trotignon, \cite{nicolas:bsp}]
  \label{th.33}
  Let $G$ be a Berge graph. Then either $G$ is bipartite, line graph
  of bipartite, complement of bipartite, complement of line graph of
  bipartite or double split, or one of $G, \overline{G}$ is a
  path-cobipartite graph, or one of $G, \overline{G}$ is a path-double
  split graph, or one of $G, \overline{G}$ has a homogeneous 2-join,
  or one of $G, \overline{G}$ has a connected non-path 2-join, or $G$
  has a balanced skew partition.
\end{theorem}

Here, we will only use the obvious following corollary:

\begin{theorem}
  \label{th.3}
  If $G$ is in $\C{Berge}{no cutset}$, then either $G$ is in
  $\C{Berge}{basic}$ or $G$ has a connected non-path 2-join.
\end{theorem}

\begin{proof}
  Follows directly from Theorem~\ref{th.33} and the fact that a graph
  with a homogeneous 2-join, or whose complement has a homogeneous
  2-join, admits a homogeneous pair.
\end{proof}

Note that since we need to use non-path 2-joins, we really need to
exclude homogeneous pairs (or to find a different approach).  Indeed,
there exist Berge graphs that are decomposable only with path 2-joins
and homogeneous pairs.  An example from~\cite{nicolas:bsp} is shown
Figure~\ref{fig:contrex3}.

\begin{figure}
  \begin{center}
    \includegraphics{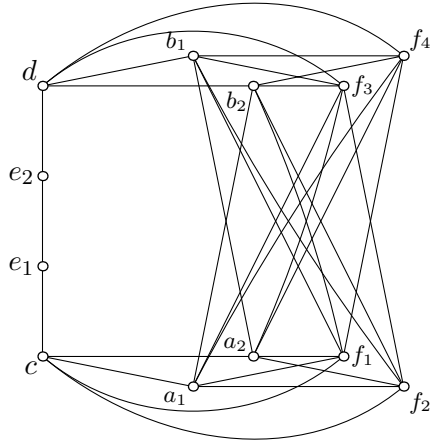}
    \caption{A graph that has a homogeneous 2-join $(\{a_1, a_2\},$
      $\{b_1, b_2\},$ $\{c\}, \{d\},$ $\{e_1, e_2\},$ $\{f_1, f_2,$ $f_3,
      f_4\})$\label{fig:contrex3}}
  \end{center}
\end{figure}

\section{Blocks of decomposition with respect to a 2-join}
\label{sec:blocks}

Blocks of decomposition with respect to a 2-join are built by
replacing each side of the 2-join by a path and the lemma below shows
that for graphs in $\C{parity}{}$ there exists a unique way to choose
the parity of that path.

\begin{lemma}
  \label{l.2jAiBi}
  Let $G$ be a graph in $\C{parity}{}$ and $(X_1,X_2,A_1,B_1,A_2,B_2)$
  be a split of a connected 2-join of $G$.  Then for $i=1,2$, all the
  paths with an end in $A_i$, an end in $B_i$ and interior in $C_i$
  have the same parity.
\end{lemma}

\begin{proof}
  Since $(X_1, X_2)$ is connected there exits a path $P$ with one end
  in $A_{3-i}$, one end in $B_{3-i}$ and interior in $C_{3-i}$.  If
  two paths $Q, R$ from $A_i$ to $B_i$ with interior in $C_i$ are of
  different parity then the holes $P\cup Q$ and $P\cup R$ are of
  different parity, a contradiction to $G\in \C{parity}{}$. 
\end{proof}

\label{sec.defpiece}
Let $G$ be a graph and $(X_1, X_2, A_1, B_1, A_2, B_2)$ be a split of
a connected 2-join of $G$.  Let $k_1, k_2 \geq 1$ be integers.  The
\emph{blocks of decomposition} of $G$ with respect to $(X_1, X_2)$ are
the two graphs $G^{k_1}_1, G^{k_2}_2$ that we describe now.  We obtain
$G^{k_1}_1$ by replacing $X_2$ by a \emph{marker path} $P_2$, of
length $k_1$, from a vertex $a_2$ complete to $A_1$, to a vertex $b_2$
complete to $B_1$ (the interior of $P_2$ has no neighbor in $X_1$).
The block $G_2^{k_2}$ is obtained similarly by replacing $X_1$ by a
marker path $P_1$ of length $k_2$.  We say that $G^{k_1}_1$ and
$G^{k_2}_2$ are {\em parity-preserving} if $G$ is in $\C{parity}{}$, for
$i=1,2$ and for a path $Q_i$ from $A_i$ to $B_i$ whose intermediate
vertices are in $C_i$ (and such a path exists since $(X_1,X_2)$ is
connected), the marker path $P_i$ has the same parity as $Q_i$.  Note
that by Lemma~\ref{l.2jAiBi}, our definition does not depend on the
choice of a particular $Q_i$.

\subsection{Interaction between 2-joins and cutsets}
\label{sec:inter}

Here we show that assuming that a graph does not admit the cutsets we
consider gives several interesting properties for its 2-joins.

\begin{lemma}\label{k1}
  Let $G$ be in $\C{}{no sc}$ and let $(X_1, X_2, A_1, B_1, A_2, B_2)$
  be a split of a 2-join of $G$.  Then the following hold:
  \begin{enumerate}
  \item\label{i:OldConn}
    every component of $G[X_i]$ meets both $A_i$ and $B_i$, $i=1,2$;
  \item
    \label{i:con} 
    $(X_1,X_2)$ is connected;
  \item
    \label{i:neiXi}  
    every $u \in X_i$ has a neighbor in $X_i$, $i=1,2$;
  \item
    \label{i:BiComp1}
    every vertex of $A_i$ has a non-neighbor in $B_i$, $i=1,2$;
  \item
    \label{i:AiComp1}
    every vertex of $B_i$ has a non-neighbor in $A_i$, $i=1,2$;
  \item
    \label{i:geq4} 
    $|X_i|\geq 4$,  $i=1,2$.
\end{enumerate}
\end{lemma}

\begin{proof}
  To prove~\ref{i:OldConn}, suppose for a contradiction that some
  connected component $C$ of $G[X_1]$ does not intersect $B_1$ (the
  other cases are symmetric).  If there is a vertex $c\in C\sm A_1$
  then for any vertex $u\in A_2$, we have that $\{u\}\cup A_1$ is a
  star cutset that separates $c$ from $B_1$.  So, $C\subseteq A_1$.
  If $|A_1| \geq 2$ then pick any vertex $c\in C$ and a vertex $c'\neq
  c$ in $A_1$.  Then $\{c'\} \cup A_2$ is a star cutset that separates
  $c$ from $B_1$.  So, $C= A_1 = \{c\}$.  Hence, there exists some
  component of $G[X_1]$ that does not intersect $A_1$, so by the same
  argument as above we deduce $|B_1| = 1$ and the unique vertex of
  $B_1$ has no neighbor in $X_1$.  Since $|X_1|\geq 3$, there is a
  vertex $u$ in $C_1$.  For any vertex $v$ in $X_2$, $\{v\}$ is a star
  cutset of $G$ that separates $u$ from $A_1$, a contradiction.

  Item \ref{i:con} follows directly from~\ref{i:OldConn}.

  To prove~\ref{i:neiXi}, just notice that if some vertex in $X_i$ has
  no neighbor in $X_i$, then it forms a component of $G[X_i]$ that
  does not intersect one of $A_i, B_i$.  This is a contradiction
  to~\ref{i:OldConn}.

  To prove~\ref{i:BiComp1} and~\ref{i:AiComp1}, consider a vertex
  $a\in A_1$ complete to $B_1$ (the others cases are symmetric).  If
  $A_1\cup C_1 \neq \{a\}$ then $B_1\cup A_2 \cup \{a\}$ is a star
  cutset that separates $(A_1\cup C_1) \sm \{a\}$ from $B_2$, a
  contradiction.  So, $A_1\cup C_1 = \{a\}$ and $|B_1| \geq 2$ because
  $X_1\geq 3$.  Let $b\neq b' \in B_1$.  So, $\{b, a\} \cup B_2$ is a
  star cutset that separates $b'$ from $A_2$, a contradiction.

  To prove~\ref{i:geq4}, suppose for a contradiction that $|X_1| = 3$.
  Up to symmetry we assume $|A_1| = 1$, and let $a_1$ be the unique
  vertex in $A_1$.  As we just proved, every vertex of $B_1$ has a
  non-neighbor in $A_1$.  Since $A_1 = \{a_1\}$, this means that $a_1$
  has no neighbor in $B_1$.  Since $(X_1, X_2)$ is a connected 2-join
  (because of~\ref{i:con}), $G[X_1]$ must be path of length 2 whose
  interior is in $C_1$.  This contradicts the definition of a 2-join.
\end{proof}

Note that a star cutset of size at least~2 is a skew cutset.  The
following was noticed by Zambelli~\cite{zambelli:these} and is
sometimes very useful.  A proof can be found in~\cite{nicolas:bsp},
Lemma~4.3.

\begin{lemma}
  \label{l.starcutset}
  Let $G$ be a Berge graph of order at least~5 with at least one
  edge. If $G$ has a star cutset then $G$ has a balanced skew
  partition.
\end{lemma}

\begin{lemma}
  \label{l:nostar}
  If a graph $G$ is in $\C{Berge}{no bsp}$ and admits a 2-join or a
  homogeneous pair, then neither $G$ nor $\overline{G}$ has a star
  cutset.
\end{lemma}

\begin{proof}
  Assume the hypothesis.  Since $G$ has 2-join or a homogeneous pair,
  it is of order at least~5 and both $G$ and $\overline{G}$ have at
  least one edge.  So, by Lemma~\ref{l.starcutset}, if $G$ has a star
  cutset, then $G$ has a balanced skew partition, a contradiction.
  Similarly, by Lemma~\ref{l.starcutset}, if $\overline{G}$ has a star
  cutset, then $\overline{G}$ has a balanced skew partition, and hence
  so does $G$, a contradiction.
\end{proof}

\begin{lemma}
  \label{l:degenerate}
  Let $G \in \C{ehf}{no sc} \cup \C{Berge}{no bsp}$.  If $(X_1, X_2,
  A_1, B_1, A_2, B_2)$ is a split of a 2-join of $G$ then every vertex
  of $A_i$ has a neighbor in $X_i\sm A_i$, $i=1,2$; and every vertex
  of $B_i$ has a neighbor in $X_i\sm B_i$, $i=1,2$.
\end{lemma}

\begin{proof}
  By Lemma \ref{l:nostar}, $G$ has no star cutset, so by Lemma
  \ref{k1} \ref{i:con}, $(X_1,X_2)$ is connected.  Consider a vertex
  $a\in A_1$ with no neighbor in $X_1\sm A_1$ (the other cases are
  similar).  We put $Z = (A_1\cup A_2)\sm \{a\}$.  So, $Z$ is a cutset
  that separates $a$ from the rest of the graph.  We note that
  $A_1\neq \{a\}$ because $(X_1, X_2)$ is connected.  If $G\in
  \C{Berge}{no bsp}$ then $V(G)\sm Z$ is a star cutset (centered at
  $a$) of $\overline{G}$ and this contradicts Lemma~\ref{l:nostar}.
  If $G\in\C{ehf}{no sc}$ then we note that at least one of $A_1, A_2$
  is a clique because $G$ contains no 4-hole.  So $Z$ is a star cutset
  of ${G}$, a contradiction.
\end{proof}

\subsection{Staying in the class}
\label{sec:stay}

Here we prove several lemmas of the same flavour, needed later for
inductive proofs and recursive algorithms.  They all say that building
the blocks of a graph with respect to some well chosen 2-join
preserves several properties like being free of cutset or member of
some class.

\begin{lemma}
  \label{l:preservesCycle}
  Let $G$ be a graph in $\C{parity}{no sc}$ and $(X_1, X_2)$ a
  connected 2-join of $G$.  Let $G_1^{k_1}$, $G_2^{k_2}$ be
  parity-preserving blocks of decomposition of $G$ w.r.t.\
  $(X_1,X_2)$, where $k_1, k_2 \geq 2$.  If one of $G_1^{k_1}$,
  $G_2^{k_2}$ contain an odd (resp.\ even) hole, then $G$ contains an
  odd (resp.\ even) hole.
\end{lemma}

\begin{proof}
  Let $C$ be a hole in $G_1^{k_1}$ say.  Let $P_2 = a_2 \tp \cdots \tp
  b_2$ be the marker path of $G_1^{k_1}$.  If $C \subseteq X_1 \cup
  \{a_2, b_2\}$ then we obtain a hole $C'$ of $G$ as follows.  By
  Lemma~\ref{k1} \ref{i:BiComp1}, there exist non-adjacent vertices
  $a'_2\in A_2$, $b'_2 \in B_2$.  If $a_2 \in C$ (resp.\ if $b_2 \in
  C$) then we replace $a_2$ (resp.\ $b_2$) by $a'_2$, (resp.\ $b'_2$).
  So, holes $C$ and $C'$ have the same length and in particular the
  same parity.

  So we may assume that $C$ contains interior vertices of $P_2$.  This
  means that $C$ is the union of $P_2$ together with a path $Q$ from
  $A_1$ to $B_1$ with interior in $C_1$.  We obtain a hole $C'$ of $G$
  by replacing $P_2$ by any path of $G$ from $A_2$ to $B_2$ with
  interior in $C_2$.  Such a path exists because $(X_1, X_2)$ is
  connected.  Holes $C$ and $C'$ have the same parity from the
  definition of parity-preserving blocks and Lemma \ref{l.2jAiBi}.
\end{proof}

\begin{lemma}
  \label{l:preservesNoSC}
  Let $G\in \C{ehf}{no sc} \cup
  \C{Berge}{no bsp}$ and let $(X_1, X_2)$
  be a connected 2-join of~$G$.  Let
  $G^{k_1}_1$ and $G^{k_2}_2$ be blocks of decomposition w.r.t.\
  $(X_1,X_2)$.  If $k_1,k_2\geq 3$ then $G^{k_1}_1$ and $G^{k_2}_2$
  are both in $\C{}{no sc}$.
\end{lemma}

\begin{proof}
  By Lemma~\ref{l:nostar}, $G$ is in $\C{}{no sc}$.  Let $(X_1, X_2,
  A_1,B_1,A_2,B_2)$ be a split of $(X_1,X_2)$.  Let $P_2 = a_2 \tp
  \cdots \tp b_2$ be the marker path of $G^{k_1}_1$.

  Suppose that $G_1^{k_1}$ say has a star cutset $S$ centered at
  $x$. If $S \cap P_2 =\emptyset$ then $S$ is a star cutset of $G$, a
  contradiction.  So $S \cap P_2 \neq \emptyset$. If $x \not\in P_2$
  then w.l.o.g.\ $x \in A_1$ and hence $S \cup A_2$ is a star cutset
  of $G$, a contradiction.

  So $x \in P_2$. First suppose that $x$ coincides with $a_2$ or
  $b_2$, say $x=a_2$.  Since $k_1 >1$, vertices of $B_1\cup \{b_2\}$ are
  all contained in the same component $B$ of $G_1^{k_1} \setminus S$.
  Let $C$ be a connected component of $G_1^{k_1} \setminus S$ that is
  distinct from $B$.  If $C \setminus A_1 \neq \emptyset$ then for
  $a_2' \in A_2$, $A_1 \cup \{ a_2' \}$ is a star cutset of $G$, a
  contradiction.  So $C \subseteq A_1$.  Hence, some vertex $c$ of $C$
  is in $A_1$ and has no neighbor in $X_1\sm A_1$, a contradiction
  to Lemma~\ref{l:degenerate}.
  
  Therefore, $x \in P_2 \setminus \{ a_2,b_2 \}$.  Since $(X_1,X_2)$
  satisfies~\ref{i:OldConn} from Lemma~\ref{k1}, both $G[X_1 \cup \{
  a_2\}]$ and $G[X_1 \cup \{b_2\}]$ are connected.  So, both $a_2$ and
  $b_2$ must be in $S$, a contradiction to $k_1\geq 3$.
\end{proof}

\begin{lemma}\label{k:even}
  Let $G\in \C{ehf}{no sc}$.  Let $(X_1, X_2)$ be a connected 2-join
  of $G$ and let $G_1^{k_1}$ and $G_2^{k_2}$ be parity-preserving
  blocks of decomposition.  If $k_1,k_2 \geq 3$ then $G_1^{k_1}$ and
  $G_2^{k_2}$ are both in $\C{ehf}{no sc}$.
\end{lemma}

\begin{proof}
  For $i=1, 2$, by Lemma~\ref{l:preservesNoSC}, $G_i^{k_i}$ has no
  star cutset.  By Lemma~\ref{l:preservesCycle}, $G_i^{k_i}$ contains
  no even hole.
\end{proof}

\begin{lemma}
  \label{l:bergeBSP}
  Let $G\in \C{Berge}{no bsp}$ and let $(X_1, X_2)$ be a connected
  non-path 2-join of $G$. Let $G^{k_1}_1$ and $G^{k_2}_2$ be
  parity-preserving blocks of decomposition of $G$ w.r.t.\ $(X_1,
  X_2)$ where $3 \leq k_1, k_2 \leq 4$.  Then $G^{k_1}_1$ and
  $G^{k_2}_2$ are both in $\C{Berge}{no bsp}$.
\end{lemma}

\begin{proof}
  Lemma 4.12 and~4.18 from~\cite{nicolas:bsp} say that if $(X_1, X_2)$
  is \emph{proper} then the conclusion holds, where a \emph{proper
    2-join} means a 2-join that satisfies~\ref{i:OldConn} from
  Lemma~\ref{k1}.  So, by Lemma~\ref{k1}, $(X_1, X_2)$ is proper and
  the conclusion holds.
\end{proof}

\begin{lemma}
  \label{l:HPdegenerate}
  Let $G \in \C{Berge}{no bsp}$ and let $(A, B, C, D, E, F)$ be a
  homogeneous pair of $G$.  Then every vertex of $C$ has a neighbor in
  $E\cup D$ and every vertex of $D$ has a neighbor in $E\cup C$.
\end{lemma}

\begin{proof}
  If there exists a vertex $c\in C$ with no neighbor in $E\cup D$ then
  $(A \cup C \cup F) \sm \{c\}$ is a skew cutset that separates $c$
  from the rest of the graph.  Hence, $\overline{G}$ has a star-cutset
  centered at $c$, a contradiction to Lemma~\ref{l:nostar}.  The
  case with $d\in D$ is similar.
\end{proof}

Berge graphs have a particular problem: their decomposition theorems
allow 2-joins in the complement.  And swapping to the complement makes
difficult keeping track of maximum stable sets and cliques.  The
following lemma shows how to bypass this problem.

\begin{lemma}
  \label{l:recurseBerge}
  Let $G\in \C{Berge}{no cutset}$ and let $(X_1, X_2)$ be a connected
  non-path 2-join of $G$.  Let $G^{k_1}_1$ and $G^{k_2}_2$ be
  parity-preserving blocks of decomposition of $G$ w.r.t.\ $(X_1,
  X_2)$ where $3\leq k_1, k_2 \leq 4$.  Then $G^{k_1}_1$ and
  $G^{k_2}_2$ are in $\C{Berge}{no cutset}$.  
\end{lemma}

\begin{proof}
  Note that $G$ is in $\C{Berge}{no bsp}$.  By Lemma~\ref{l:bergeBSP},
  $G^{k_1}_1$ and $G^{k_2}_2$ are both Berge graphs and $G^{k_1}_1$
  and $G^{k_2}_2$ have no balanced skew partition.  Because of the
  symmetry, we just have to prove the following two claims:

  \begin{claim}
    $\overline{G^{k_2}_2}$ has no connected 2-join.  
  \end{claim}
  
  \begin{proofclaim}
    Note that $G^{k_2}_2$ has a 2-join (a path 2-join).  So by
    Lemma~\ref{l:nostar}, $\overline{G^{k_2}_2}$ has no star cutset.
    In $G^{k_2}_2$, we denote by $P_1 = a_1 \tp c_1 \tp c_2 \tp c' \tp
    b_1$ the marker path, where $a_1$ is complete to $A_2$ and $b_1$
    to $B_2$.  Note that this path may be of length 3 or 4.  If it is
    of length 3, then we suppose $c'=b_1$, this is convenient to avoid
    a multiplication of cases.  Note that by Lemma
    \ref{k1} \ref{i:con}, any 2-join of $\overline{G_2^{k_2}}$ is
    connected.   Let us suppose for a contradiction
    that $\overline{G^{k_2}_2}$ has a 2-join.   Let $(X'_1, X'_2, A'_1, B'_1, A'_2, B'_2)$ be a split
    of this 2-join.  We put $C'_i = X'_i\sm (A'_i \cup B'_i)$, $i=1,
    2$.
    
    Since $c_1, c_2$ are of degree $2$ in $G^{k_2}_2$, they are of
    degree $|V(\overline{G_2^{k_2}})|-3$ in $\overline{G^{k_2}_2}$.
    Since by Lemma~\ref{k1} \ref{i:geq4}, $|X'_i|\geq 4$, $i= 1, 2$,
    we must have $c_1, c_2 \in A'_1 \cup A'_2 \cup B'_1 \cup B'_2$.
    Also, $c_1, c_2$ are non-adjacent in $\overline{G^{k_2}_2}$ so up
    to symmetry there are two cases.  Here below, the words
    ``neighbor'' and ``non-neighbor'' refer to adjacency in
    $\overline{G^{k_2}_2}$.

    Case 1: $c_1\in A'_1$, $c_2 \in B'_2$. Then by Lemma~\ref{k1},
    \ref{i:BiComp1} applied to $(X'_1, X'_2)$, $c_1$ must have a
    non-neighbor in $B'_1$, and this non-neighbor must be $a_1$ ($a_1,
    c_2$ are the only non-neighbors of $c_1$).  So, $a_1 \in B'_1$.
    Similarly, $c_2$ must have a non-neighbor in $A_2'$, and this
    non-neighbor must be $c'$, i.e.\ $c' \in A'_2$.  But then, since
    $a_1c'$ is an edge of $\overline{G_2^{k_2}}$, this contradicts the
    definition of a 2-join.

    Case 2: $c_1\in A'_1$, $c_2 \in B'_1$.  We must have $B'_2 =
    \{a_1\}$ because $a_1, c_2$ are the only non-neighbors of $c_1$.
    For the same reason, $A'_2 = \{c'\}$.  Since $a_1$ is a neighbor
    of $c'$, there is a contradiction to Lemma~\ref{k1}
    \ref{i:BiComp1} applied to $(X'_1, X'_2)$.
  \end{proofclaim}

  \begin{claim}
    $G^{k_2}_2$ has no homogeneous pair.      
  \end{claim}
  
  \begin{proofclaim}
    Suppose for a contradiction that $(A, B, C, D, E, F)$ is a
    homogeneous pair of $G^{k_2}_2$.  It is convenient to use now a
    slightly different notation for the marker path $P_1 = a_1 \tp c_1
    \tp c_2 \tp c' \tp b_1$.  Now, when the path is of length 3, we
    suppose $c'=c_2$.

    By Lemma~\ref{l:HPdegenerate}, every vertex in $A, B, C, D, F$ has
    degree at least 3.  So, $c_1, c_2, c' \in E$.  We will reach the
    contradiction by giving an homogeneous pair $(A', B', C', D', E',
    F')$ of $G$. 

    In $E'$, we put every vertex of $E\sm \{c_1, c_2, c'\}$, and we
    add $C_1$.  If $E$ contains $a_1$, we add $A_1$ to $E'$.  If $E$
    contains $b_1$, we add $B_1$ to $E'$.  We set $A'=A$ and $B'=B$.
    In $C'$, we put every vertex of $C\sm\{a_1, b_1\}$.  If $C$
    contains $a_1$, we add $A_1$ to $C'$.  If $C$ contains $b_1$, we
    add $B_1$ to $C'$.  We define similarly $D'$, $F'$ from $D$ and
    $F$ respectively.  Now we observe that $(A', B', C', D', E', F')$
    partitions $V(G)$ and is a homogeneous pair of $G$ (note that $E'$
    is possibly empty).
  \end{proofclaim}
  \mbox{}
\end{proof}

\section{Extreme 2-joins}
\label{sec:extrem}

In this section, we show how to find a 2-join in a graph so that one
of the blocks has no more 2-joins.  The idea behind this is that when
$(X_1, X_2)$ is a 2-join of $G$ and the block of decomposition
$G_1^{k_1}$ has no more decomposition then it is basic.  So a lot of
computations can be made on $G[X_1]$ without leading to an exponential
complexity.

Let $(X_1,X_2)$ be a connected non-path 2-join of a graph $G$.  We say
that $(X_1,X_2)$ is a {\em minimally-sided connected non-path 2-join}
if for some $i\in \{ 1,2 \}$, the following holds: for every connected
non-path 2-join $(X_1',X_2')$ of $G$, neither $X_1' \subsetneq X_i$
nor $X'_2 \subsetneq X_i$ holds.  We call $X_i$ a {\em minimal side}
of this minimally-sided 2-join.  Note that minimally-sided connected
non-path 2-joins exist in any graph that admits a connected non-path
2-join.

Let $(X_1,X_2)$ be a connected non-path 2-join of $G$. We say that
$(X_1,X_2)$ is an {\em extreme 2-join} if for some $i\in \{ 1,2 \}$
and all $k \geq 3$ the block of decomposition $G_i^{k}$ has no
connected non-path 2-join.  We say that $X_i$ is an \emph{extreme
  side} of such a 2-join.  Figure~\ref{fige2j} shows that graphs in
general do not have an extreme 2-join, but as we now show, graphs with
no star cutset do.

  \begin{figure}[h!]
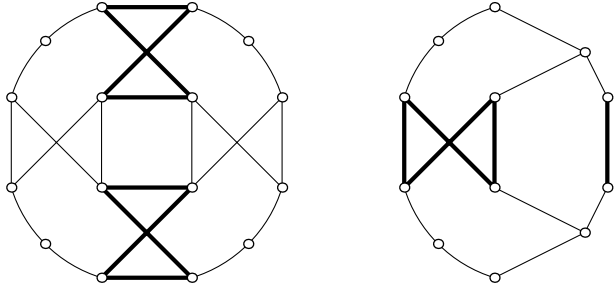

   \begin{center}
     \includegraphics{lineExt.6}   \rule{3em}{0ex}  \includegraphics{lineExt.7}
\end{center}
\caption{\label{fige2j} Graph $G$ (on the left) has a star
  cutset, but does not have an extreme connected non-path 2-join: $G$
  has a connected non-path 2-join represented with bold lines, and all
  connected non-path 2-joins are equivalent to this one. Both of the
  blocks of decomposition are isomorphic to graph $H$ (on the right),
  and $H$ has a connected non-path 2-join whose edges are represented
  with bold lines.}
\end{figure}

\begin{lemma}\label{k21}
  Let $G \in \C{}{no sc}$ and let $(X_1,X_2,A_1,B_1,A_2,B_2)$ be a
  split of a connected non-path 2-join of $G$.  If $|A_1|=1$ then
  $(X_1 \setminus A_1,X_2 \cup A_1,N_{G[X_1]}(A_1),B_1,A_1,B_2)$ is a
  split of a connected non-path 2-join of $G$.
\end{lemma}

\begin{proof}
  Assume $A_1=\{ a_1 \}$.  So by Lemma~\ref{k1} \ref{i:AiComp1}, $a_1$
  has no neighbor in $B_1$.  Let $A_1'=N_{G[X_1]}(a_1)$. Then $A_1'
  \cap B_1=\emptyset$.  By Lemma \ref{k1} \ref{i:geq4}, $|X_1
  \setminus A_1| \geq 3$.  Since $(X_1,X_2)$ is a non-path 2-join of
  $G$, it follows that $(X_1 \setminus A_1,X_2 \cup A_1, A_1', B_1,
  A_1, B_2)$ is a split of a non-path 2-join of $G$.  By Lemma
  \ref{k1} \ref{i:con}, $(X_1 \setminus A_1,X_2 \cup A_1)$ is
  connected.
\end{proof}

\begin{lemma}\label{k22}
  Let $G\in \C{}{no sc}$ and let $(X_1,X_2,A_1,B_1,A_2,B_2)$ be a
  split of a minimally-sided connected non-path 2-join of $G$. 
Let $X_1$ be a minimal side.
Then
  $|A_1| \geq 2$ and $|B_1| \geq 2$.  In particular, $A_2 \cup B_2$
  contains only vertices of degree at least 3.
\end{lemma}

\begin{proof}
  If $|A_1|=1$ then by Lemma \ref{k21}, $(X_1 \setminus A_1,X_2 \cup
  A_1)$ is a connected non-path 2-join of $G$, contradicting the
  assumption that $(X_1,X_2)$ is minimally sided connected non-path
  2-join of $G$.  So $|A_1| \geq 2$, and by symmetry $|B_1| \geq 2$.
  By Lemma \ref{k1} \ref{i:neiXi}, all vertices of $A_2\cup B_2$ have
  degree at least 3.
\end{proof}

Recall that a {\em flat path} of a graph $G$ is any path of $G$ of length
at least~2, whose interior vertices all have degree 2 in $G$, and whose
ends have no common neighbors outside the path.

\begin{lemma}\label{l:pathThrough}
  Let $G\in \C{}{no sc}$ and let $(X_1,X_2,A_1,B_1,A_2,B_2)$ be a
  split of a minimally sided connected non-path 2-join of $G$.  Let
  $X_1$ be a minimal side, and let $P$ be a flat path of $G$.  If $P
  \cap X_1 \neq \emptyset$ and $P \cap X_2 \neq \emptyset$, then
  one of the following holds:
  \begin{enumerate}
  \item\label{i:onev} For an endvertex $u$ of $P$, $P \setminus u
    \subseteq X_1$ and $u \in A_2 \cup B_2$.
  \item\label{i:twov} For endvertices $u$ and $v$ of $P$, $u \in A_2$, $v
    \in B_2$, $P \setminus \{ u,v \} \subseteq X_1$, the length of $P$
    is at least 3 and $G[X_1]$ has exactly two connected components
    that are both a path with one end in $A_1$, one end in $B_1$ and
    interior in $C_1$.
  \end{enumerate}
\end{lemma}

\begin{proof}
  Let $u$ and $v$ be the endvertices of $P$.  By Lemma~\ref{k22}, the
  interior of $P$ must lie in $X_1$ and w.l.o.g.\ $u \in A_2$.  By
  Lemma \ref{k1} \ref{i:neiXi}, the neighbor $x$ of $u$ along $P$ has
  a neighbor in $X_1$, so $|A_2|=1$.  Hence if \ref{i:onev} does not
  hold then $v \in B_2$.  Also, $P$ is of length at least 3.  So, the
  interior of $P$ is a connected component $C$ of $G[X_1]$.  If
  $G[X_1\sm C]$ is not a path with one end in $A_1$, one end in $B_1$
  and interior in $C_1$ then $(X_1\sm C, X_2 \cup C)$ is a connected
  non-path 2-join.  Indeed, by Lemma~\ref{k1}~\ref{i:OldConn}, $X_1\sm
  C$ meets $A_1$ and $B_1$, and since it is not a path, it has size at
  least 3, so $(X_1\sm C, X_2 \cup C)$ is a non-path 2-join that is
  connected by Lemma~\ref{k1}~\ref{i:con}.  This contradicts $(X_1,
  X_2)$ being minimally-sided.  Therefore \ref{i:twov} holds.
\end{proof}

\begin{lemma}\label{k4}
  Let $(X_1, X_2, A_1, B_1, A_2, B_2)$
  be a split of a minimally-sided connected non-path 2-join of a graph $G$,
with $X_1$ being a minimal side.
Assume that $G$ and all the blocks of decomposition of $G$ w.r.t.\ 
$(X_1,X_2)$ whose marker paths are of length at least 3, all belong to
$\C{}{no sc}$. 
 Then $(X_1,X_2)$ is an
  extreme 2-join and $X_1$ is an extreme side.
\end{lemma}

\begin{proof}
  Suppose that the block of decomposition $G_1^{k_1}$, $k_1 \geq 3$, has a
  connected non-path 2-join with split $(X_1', X_2', A_1', B_1', A_2',
  B_2')$.  For $i=1,2$ let $C_i'=X_i' \setminus (A_i' \cup B_i')$.
  Let $P_2 = x_0 \tp x_1 \tp \cdots \tp x_{k_1}$, where $x_0=a_2$ and
  $x_{k_1}=b_2$, be the marker path of $G_1^{k_1}$.

\vspace{2ex}

\noindent
{\bf Case 1:} For some $i \in \{ 1,2 \}$, $P_2 \subseteq X_i'$.

\vspace{1ex}

W.l.o.g.\ $P_2 \subseteq X_2'$. Note that since $N_{G_1^{k_1}} (P_2
\setminus \{ a_2,b_2 \}) \subseteq \{ a_2,b_2 \}$ we have $P_2 \cap
(A_2' \cup B_2') \subseteq \{ a_2,b_2 \}$. Note also that since $a_2$
and $b_2$ have no common neighbor in $G_1^{k_1}$ we have $|A_2' \cap
\{ a_2,b_2 \}|\leq 1$ and $|B_2' \cap \{ a_2,b_2 \}|\leq 1$. So by
symmetry it suffices to consider the following subcases:

\vspace{2ex}

\noindent
{\bf Case 1.1:} $P_2 \subseteq C_2'$

\vspace{1ex}

Since $a_2$ is adjacent to all vertices of $A_1$ and $a_2$ has no neighbor
in $X_1'$ (since $a_2 \not\in A_2' \cup B_2'$) it follows that 
$A_1 \subseteq X_2'$. Similarly $B_1 \subseteq X_2'$. But then 
$(X_1', (X_2' \setminus P_2) \cup X_2)$ is a connected non-path 2-join
of $G$. Since $A_1 \cup B_1 \subseteq X_2'$, $X_1' \subsetneq X_1$, 
contradicting our choice of $(X_1,X_2)$.

\vspace{2ex}

\noindent
{\bf Case 1.2:} $a_2 \in A_2'$ and $P_2 \setminus \{ a_2\} \subseteq C_2'$.

\vspace{1ex}

So $b_2$ has no neighbor in $X_1'$, and since $b_2$ is adjacent to all
vertices of $B_1$, it follows that $B_1 \subseteq X_2'$.  In particular,
$X_1' \subsetneq X_1$.  Since $a_2 \in A_2'$, $P_2 \subseteq X_2'$ and
$N_{G_1^{k_1}} (a_2) \setminus P_2 =A_1$, it follows that $A_1' \subseteq
A_1$ and $(X_1' \setminus A_1') \cap A_1=\emptyset$. But then $(X_1',
(X_2' \setminus P_2) \cup X_2,A_1',B_1',(A_2'\setminus \{ a_2\}) \cup
A_2,B_2')$ is a split of a connected non-path 2-join of $G$,
contradicting our choice of $(X_1,X_2)$.

\vspace{2ex}

\noindent
{\bf Case 1.3:} $a_2 \in A_2'$, $b_2 \in B_2'$ and 
$P_2 \setminus \{ a_2,b_2 \} \subseteq C_2'$.

\vspace{1ex}

Since $(X_1',X_2')$ is not a path 2-join, $X_2' \cap X_1 \neq
\emptyset$, and hence $X_1' \subsetneq X_1$.  Since $a_2 \in A_2'$, $P_2
\subseteq X_2'$ and $N_{G_1^{k_1}} (a_2) \setminus P_2 =A_1$, it follows
that $A_1' \subseteq A_1$ and $(X_1' \setminus A_1') \cap A_1
=\emptyset$.  Similarly $B_1' \subseteq B_1$ and $(X_1' \setminus
B_1') \cap B_1 =\emptyset$.  But then $(X_1', (X_2' \setminus P_2)
\cup X_2,A_1',B_1',(A_2'\setminus \{ a_2\}) \cup A_2,(B_2'\setminus \{
b_2 \}) \cup B_2)$ is a split of a connected non-path 2-join of $G$,
contradicting our choice of $(X_1,X_2)$.

\vspace{2ex}

\noindent
{\bf Case 2:} For $i = 1,2$, $P_2 \cap X_i'\neq \emptyset$.

\vspace{1ex}

We may assume w.l.o.g.\ that $(X_1',X_2')$ is a minimally-sided
connected non-path 2-join of $G_1^{k_1}$, with $X_1'$ being the
minimal side.  But then by Lemma \ref{l:pathThrough} applied to
$(X_1',X_2')$ and $P_2$ it suffices up to symmetry to consider the
following two cases:

\vspace{2ex}

\noindent
{\bf Case 2.1:} $a_2 \in A_2'$ and $P_2 \setminus a_2 \subseteq X_1'$.

\vspace{1ex}

Since $x_1$ is of degree 2 and $x_2 \in X_1'$, it follows that
$|A_2'|=1$. Note that $P_2 \setminus \{ a_2,x_1\} \subseteq X_1'
\setminus A_1'$.  Since $a_2 \in A_2'$, $A_1' \setminus \{ x_1\}
\subseteq A_1$ and $A_1 \setminus A_1' \subseteq X_2'\setminus A_2'$.
Note that since $(X_1',X_2')$ is connected, $A_1 \setminus A_1' \neq
\emptyset$.  Note that $P_2 \cap B_1'\subseteq \{ b_2\}$.

First suppose that $b_2 \in C_1'$. Since $b_2$ is adjacent to all
vertices of $B_1$, $B_1 \subseteq X_1'$.  By Lemma~\ref{k1},
\ref{i:AiComp1}, no vertex of $B'_2$ has a neighbor in $A'_2$, which
implies $B_2' \cap A_1=\emptyset$.  So, $C_2' \neq \emptyset$ and by
Lemma \ref{k1} \ref{i:geq4}, $|X_2'| \geq 4$.  But then
$(X_2'\setminus \{ a_2\}, (X_1' \setminus P_2) \cup X_2, A_1\setminus
A_1',B_2',A_2,B_1')$ is a split of a connected non-path 2-join of $G$,
contradicting our choice of $(X_1,X_2)$ (since clearly $X_2'
\setminus \{ a_2 \} \subsetneq X_1$).

Hence $b_2 \not \in C_1'$ and  $b_2 \in B_1'$. Then $B_2'
\subseteq B_1$ and $B_1 \setminus B_2' \subseteq X_1'$. By
Lemma~\ref{k1}, \ref{i:AiComp1}, no vertex of $B'_2$ has a neighbor in
$A'_2$, which implies $B_2' \cap A_1 =\emptyset$.  So $(X_2'\setminus
\{ a_2\}, (X_1' \setminus P_2) \cup X_2, A_1\setminus
A_1',B_2',A_2,(B_1'\setminus \{ b_2 \}) \cup B_2)$ is a split of a
connected non-path 2-join of $G$, contradicting our choice of
$(X_1,X_2)$ (since clearly $X_2' \setminus \{ a_2 \} \subsetneq X_1$).

\vspace{2ex}

\noindent {\bf Case 2.2:} $a_2 \in A_2'$, $b_2 \in B_2'$ and $P_2
\setminus \{ a_2,b_2 \} \subseteq X_1'$.

\vspace{1ex}

Then $x_1 \in A_1'$, $x_{k_1-1} \in B_1'$ and $P \setminus \{
a_2,b_2,x_1, x_{k_1-1} \} \subseteq C_1'$.  Since $x_1$ and $x_{k_1-1}$ are
of degree 2 in $G_1^{k_1}$, it follows that $A_2'=\{ a_2 \}$ and
$B_2'=\{ b_2\}$.  Since $N_{G_1^{k_1}}(a_2)=A_1 \cup \{ x_1 \}$, it
follows that $A_1' \setminus \{ x_1 \} \subseteq A_1$ and $A_1
\setminus A_1' \subseteq X_2'$. Similarly $B_1' \setminus \{ x_{k_1-1}
\} \subseteq B_1$ and $B_1 \setminus B_1' \subseteq X_2'$.  Since
$(X_1',X_2')$ is connected, $A_1 \setminus A_1' \neq \emptyset$ and
$B_1 \setminus B_1' \neq \emptyset$.  Since $a_2$ and $b_2$ have no
common neighbor in $G_1^{k_1}$, $(A_1 \setminus A_1') \cap (B_1
\setminus B_1')=\emptyset$.  Since $(X_1',X_2')$ is a non-path 2-join,
$|X_2'\setminus \{ a_2,b_2 \}| \geq 3$.  But then $((X_1' \setminus
P_2) \cup X_2, X_2' \setminus \{ a_2,b_2 \}, A_2,B_2,A_1 \setminus
A_1',B_1 \setminus B_1')$ is a split of a 2-join of $G$.  By Lemma
\ref{k1} \ref{i:con} this 2-join is connected.  Since
$G_1^{k_1}[X_2']$ and $G[X_2]$ are not paths, this 2-join is a
non-path 2-join. But then since $X_2' \setminus \{ a_2,b_2
\} \subsetneq X_1$, our choice of $(X_1,X_2)$ is contradicted.
\end{proof}

When ${\cal M}$ is a collection of vertex-disjoint flat paths, a
2-join $(X_1, X_2)$ is \emph{{$\cal M$}-independent} if for every path
$P$ from $\cal M$ we have either $V(P) \subseteq X_1$ or $V(P) \subseteq
X_2$.
These special types of 2-joins will have a fundamental role to play when it 
comes to computing a largest stable set.

\begin{lemma}\label{k23}
  Let $(X_1,X_2,A_1,B_1,A_2,B_2)$ be a
  split of a minimally-sided connected non-path 2-join of a graph $G$, with
  $X_1$ being a minimal side.  
Assume that $G$ and all the blocks of decomposition of $G$ w.r.t.\ 
$(X_1,X_2)$ whose marker paths are of length at least 3, all belong to
$\C{}{no sc}$. 
Let ${\cal M}$ be a set of
  vertex-disjoint flat paths of length at least~3 of $G$.  
If there
  exists a path $P \in {\cal M}$ such that $P \cap A_1 \neq \emptyset$
  and $P \cap A_2 \neq \emptyset$, then let $A_1'=A_2$, and otherwise
  let $A_1'=A_1$.  If there exists a path $P \in {\cal M}$ such that
  $P \cap B_1 \neq \emptyset$ and $P \cap B_2 \neq \emptyset$, then
  let $B_1'=B_2$, and otherwise let $B_1'=B_1$.  Let $X_1'=X_1 \cup
  A_1' \cup B_1'$ and $X_2'=V(G) \setminus X_1'$.  Then the following
  hold:
  \begin{enumerate}
  \item\label{i:k23con} $(X_1',X_2')$ is a connected non-path 2-join
    of $G$.
  \item\label{i:k23indep} $(X'_1, X'_2)$ is ${\cal M}$-independent.
  \item\label{i:k23extr} $(X_1',X_2')$ is an extreme 2-join of $G$ and
    $X'_1$ is an extreme side of this 2-join.
  \end{enumerate}
\end{lemma}

\begin{proof}
  If there exists a path $P \in {\cal M}$ such that $P \cap A_1 \neq
  \emptyset$ and $P \cap A_2 \neq \emptyset$, then by Lemma
  \ref{l:pathThrough}, either for an endvertex $u$ of $P$, $P \setminus
  u \subseteq X_1$ and $u \in A_2$; or for endvertices $u$ and $v$ of
  $P$, $u \in A_2$, $v\in B_2$ and $P \setminus \{ u,v \} \subseteq
  X_1$.  Since the intermediate vertices of $P$ are of degree 2, it
  follows that $|A_2|=1$, and so \ref{i:k23con} holds by Lemma
  \ref{k21} (possibly applied twice).

  Let $(X_1',X_2',A_1',B_1',A_2',B_2')$ be the split of $(X_1',X_2')$,
  where $A_2 \in \{A_1', A_2'\}$ and $B_2 \in \{B_1', B_2'\}$.  By
  Lemma \ref{l:pathThrough}, applied to $(X_1,X_2)$, paths $P \in
  {\cal M}$ are one of the following types:

  \begin{description}
  \item[Type 1:] $P \subseteq X_1$
  \item[Type 2:] $P \subseteq X_2$
  \item[Type 3:] For an endvertex $u$ of $P$, $u \in A_2$ and $P\setminus u 
    \subseteq X_1$.
  \item[Type 4:] For an endvertex $u$ of $P$, $u \in B_2$ and $P\setminus u 
    \subseteq X_1$.
  \item[Type 5:] For endvertices $u$ and $v$ of $P$, 
    $u \in A_2$, $v \in B_2$ and $P\setminus \{ u,v \} 
    \subseteq X_1$.
  \end{description}
  
  Note that since ${\cal M}$ is a collection of vertex-disjoint paths,
  at most one path of ${\cal M}$ is of type 3 (resp.\ type 4), and if
  there exists a type 3 (resp.\ type 4) path then for every type 2 path
  $P$, $P \cap A_2=\emptyset$ (resp.\ $P \cap B_2=\emptyset$).  Also
  there is at most one type 5 path in ${\cal M}$, and if such a path
  exists there are no type 3 and 4 paths in ${\cal M}$, and for every
  type 2 path $P$ of ${\cal M}$, $P \cap A_2=\emptyset$ and $P \cap
  B_2=\emptyset$.  So by the construction of $(X_1',X_2')$ all type 1
  (resp.\ type 2) paths w.r.t.\ $(X_1,X_2)$ stay type 1 (resp.\ type 2)
  w.r.t.\ $(X_1',X_2')$, and all type 3, 4 and 5 paths
  w.r.t.\ $(X_1,X_2)$ become type 1 w.r.t.\ $(X_1',X_2')$.  Therefore
  \ref{i:k23indep} holds.

  For $k_1,k_2 \geq 3$, let $G_1^{k_1}$ and $G_2^{k_2}$ be the blocks
  of decomposition of $G$ w.r.t.\ the 2-join $(X_1,X_2)$.  By Lemma
  \ref{k4}, $G_1^{k_1}$ has no connected non-path 2-join. Let $P_2$ be
  the marker path of $G_1^{k_1}$.  Let $G_1^{'{k_1'}}$ be the block of
  decomposition of $G$ with respect to $(X'_1, X'_2)$.  Notice that
  $G_1^{'{k_1'}}$ can be obtained from $G_1^{k_1}$ by subdividing an
  edge of $P_2$ (0, 1 or 2 times), and hence by our
assumption, $G_1^{'{k_1'}}\in \C{}{no sc}$. Therefore by Lemma~\ref{k4},
  $G_1^{'{k_1'}}$ has no connected non-path 2-join and hence
  \ref{i:k23extr} holds.
\end{proof}

\begin{lemma}\label{alg:extreme}
There is an algorithm with the following specification:
\begin{description}
\item[Input:] 
A connected graph $G$ and a set ${\cal M}$ of vertex-disjoint flat
paths of $G$ of length at least 3.
\item[Output:] One of the following is returned.
 \begin{enumerate}
\item An extreme connected non-path
2-join of $G$ (with an identified extreme side) that is ${\cal M}$-independent.
\item $G$ or a block of decomposition of $G$ w.r.t.\ some 
2-join whose marker path is of length at least 3, has a star cutset.
\item $G$ has no connected non-path 2-join.
\end{enumerate}
\item[Running time:] ${\cal O} (n^3m)$
\end{description}
\end{lemma}

\begin{proof}
  First check whether $G$ has a star cutset. Note that this can be
  done in time ${\cal O} (n^3)$ as noted by
  Chv\'atal~\cite{chvatal:starcutset}: for every $x \in V(G)$, check
  whether $G \setminus N[x]$ is disconnected, and also check whether
  there exists $y \in N(x)$ such that $y$ has no neighbor in $G
  \setminus N[x]$.  If the answer to any of these is yes, then $G$ has
  a star cutset centered at $x$, and otherwise it does not. If $G$ is
  identified as having a star cutset return (ii) and stop, and
  otherwise continue.

Note that at this point in the algorithm we know that $G \in \C{}{no sc}$,
and hence by Lemma \ref{k1} \ref{i:con} any 2-join of $G$ is connected.

Run the ${\cal O}(n^3m)$-algorithm from Theorem~5.2
in~\cite{ChHaTrVu:2-join} for $G$.  This algorithm outputs a minimally
sided non-path 2-join of an input graph with no star cutset, or
certifies that the input graph has no non-path 2-join.  If this stage
of the algorithm does not find any non-path 2-join in $G$, then return
(iii) and stop.  Otherwise let $(X_1, X_2, A_1, B_1, A_2, B_2)$ be the
split of a minimally-sided connected non-path 2-join found, and
w.l.o.g. assume that $X_1$ is a minimal side. 

Let $G_1^3$ and $G_2^3$ be the blocks of decomposition of $G$ w.r.t.\ 
$(X_1,X_2)$. Check whether $G_1^3$ and $G_2^3$ have a star cutset.
If any one of them does, then return (ii) and stop.
Otherwise, we continue and we note that since $G_1^3,G_2^3 \in \C{}{no sc}$
clearly blocks of decomposition of $G$ w.r.t.\ $(X_1,X_2)$ $G_1^{k_1}$
and $G_2^{k_2}$, for any $k_1,k_2 \geq 3$, also belong to $\C{}{no sc}$.

If there exists a path $P \in {\cal M}$ such that $P \cap A_1 \neq
\emptyset$ and $P \cap A_2 \neq \emptyset$, then let $A_1'=A_2$, and
otherwise let $A_1'=A_1$.  If there exists a path $P \in {\cal M}$
such that $P \cap B_1 \neq \emptyset$ and $P \cap B_2 \neq \emptyset$,
then let $B_1'=B_2$, and otherwise let $B_1'=B_1$.  Let $X_1'=X_1 \cup
A_1' \cup B_1'$ and $X_2'=V(G) \setminus X_1'$.  Then by
Lemma~\ref{k23}, $(X_1',X_2')$ is an extreme connected non-path 2-join
(with $X_1'$ being an extreme side) that is ${\cal M}$-independent.
We return this 2-join in (i) and stop.

Clearly this algorithm can be implemented to run in time ${\cal O} (n^3m)$.
\end{proof}

Note that later when we apply the above algorithm in our main
algorithm, if the output is (ii), then by Lemma \ref{l:preservesNoSC}
we can conclude that $G$ does not belong to $\C{Berge}{no cutset}$ or
to $\C{ehf}{no sc}$.

\section{Keeping track of cliques}
\label{sec:clique}

Here we show how to find a maximum clique in a graph using 2-joins.
For the sake of induction we have to solve the weighted version of the
problem.

Through all the next sections, by graph we mean a graph with weights
on the vertices.  Weights are numbers from $K$ where $K$ means either
the set $\R_+$ of non-negative real numbers or the set $\N_+$ of non
negative integers.  The statements of the theorems will be true for
$K= \R_+$ but the algorithms are to be implemented with $K= \N_+$.
Let $G$ be a weighted graph with a weight function $w$ on $V(G)$.
When $H$ is an induced subgraph of $G$ or a subset of $V(G)$, $w(H)$
denotes the sum of the weights of vertices in $H$.  Note that we view a
graph where no weight is assigned to the vertices as a weighted graph
whose vertices have all weight~1.
Here, $\omega (G)$ denotes the weight of a maximum weighted clique of
$G$.

Let $(X_1, X_2, A_1, B_1, A_2, B_2)$ be a split of a connected 2-join
of $G$.  We define for $k\geq 3$ the \emph{clique-block} $G_2^k$ of
$G$ with respect to $(X_1, X_2)$.  It is obtained from the block
$G^k_2$ by giving weights to the vertices.  Let $P_1= a_1 \tp x_1 \tp
\cdots \tp x_{k-1} \tp b_1$ be the marker path of $G^k_2$.  We assign
the following weights to the vertices of $G^k_2$:

\begin{itemize}
\item for every $u \in X_2$, $w_{G_2^k}(u) = w_{G}(u)$;
\item $w_{G^k_2}(a_1)=\omega (G[A_1])$;
\item $w_{G^k_2}(b_1)=\omega (G[B_1])$;
\item $w_{G^k_2}(x_1)=\omega (G[X_1]) - \omega (G[A_1])$;
\item $w_{G^k_2}(x_i)=0$, for $i=2, \ldots, k-1$.
\end{itemize}

\begin{lemma}
  \label{komega}
  $\omega (G)=\omega(G_2^k)$.
\end{lemma}

\begin{proof}
  Let $K$ be a maximum weighted clique of $G$.  We show that the
  clique-block $G_2^k$ has a clique of weight $w_{G} (K)$, and hence
  $\omega (G) \leq \omega (G_2^k)$.  If $K \subseteq X_2$ then $K
  \subseteq V(G_2^k)$.  If $K \subseteq X_1$ then $\{ a_1, x_1 \}$ is a
  clique of $G_2^k$ of weight $w_{G} (K)$.  So assume that $K \cap
  X_1 \neq \emptyset$ and $K \cap X_2\neq \emptyset$.   W.l.o.g.\ $K
  \cap A_1 \neq \emptyset$ and $K \cap A_2 \neq \emptyset$, and hence
  $K \subseteq A_1 \cup A_2$.  But then $(K \setminus A_1) \cup \{ a_1
  \}$ is a clique of $G_2^k$ of weight $w_{G} (K)$. Therefore
  $\omega (G) \leq \omega (G_2^k)$.

  Now let $K$ be a maximum weighted clique of $G_2^k$.  We show that
  $G$ has a clique of weight $w_{G_2^k} (K)$, and hence $\omega
  (G_2^k) \leq \omega (G)$.  If $K \subseteq X_2$ then $K$ is a clique
  of $G$.  Suppose $K \cap P_1 =\{ a_1 \}$, and let $K'$ be a clique
  of $A_1$ whose weight is $\omega (G[A_1])$.  Then $(K \cap A_2) \cup
  K'$ is a clique of $G$ of weight $w_{G_2^k}(K)$.  So we may assume
  that $K=\{a_1, x_1\}$.  Then $w_{G_2^k} (K) = \omega (G[X_1])$, and
  $G$ has a clique of the same weight.  Therefore $\omega (G_2^k) \leq
  \omega (G)$.
\end{proof}

\section{Keeping track of stable sets}
\label{sec:alphaTrack}

Here we show how to use 2-joins to compute maximum stable sets.  This
is more difficult than cliques mainly because stable sets may
completely overlap both sides of a 2-join.  For the sake of induction
we need to put weights on the vertices.  But even with weights, there
is an issue: we are not able to compute maximum weighted stable set of
a graph assuming that some computations are done on its blocks as
defined in Section~\ref{sec:blocks}.  So we need to enlarge slightly
our blocks to encode information, and this causes some trouble.
First, the extended blocks may fail to be in the class we are working
on.  This problem will be solved in Section~\ref{sec:tree} by building
the decomposition tree in two steps.  Also in a decomposition tree
built with our unusual blocks, the leaves may fail to be basic graphs,
so computing something in the leaves of the tree is a problem
postponed to Section~\ref{sec:extension}.

Throughout this section, $G$ is a fixed graph with a weight function
$w$ on the vertices and $(X_1,X_2,A_1,B_1,A_2,B_2)$ is a split of a
2-join of~$G$.  For $i=1,2$, $C_i=X_i \setminus (A_i \cup B_i)$.  For
any graph $H$, $\alpha (H)$ denotes the weight of a maximum weighted
stable set of $H$.  We define $a = \alpha(G[{A_1 \cup C_1}])$, $b=
\alpha(G[{B_1 \cup C_1}])$, $c = \alpha(G[{C_1}])$ and $d =
\alpha(G[{X_1}])$.

\begin{lemma}
  \label{l:4cases}
  Let $S$ be a stable set of $G$ of maximum weight.  Then one and only
  one of the following holds:

  \begin{enumerate}
  \item\label{i:4c1} $S \cap A_1 \neq \emptyset$, $S \cap B_1 =
    \emptyset$, $S\cap X_1$ is a maximum weighted stable set of $G[A_1
    \cup C_1]$ and $w(S \cap X_1) = a$;
  \item\label{i:4c2} $S \cap A_1 = \emptyset$, $S \cap B_1 \neq
    \emptyset$, $S\cap X_1$ is a maximum weighted stable set of $G[B_1
    \cup C_1]$ and $w(S \cap X_1) = b$;
  \item\label{i:4c3} $S \cap A_1 = \emptyset$, $S \cap B_1 =
    \emptyset$, $S\cap X_1$ is a maximum weighted stable set of
    $G[C_1]$ and $w(S \cap X_1) = c$;
  \item\label{i:4c4} $S \cap A_1 \neq \emptyset$, $S \cap B_1 \neq
    \emptyset$, $S\cap X_1$ is a maximum weighted stable set of
    $G[X_1]$ and $w(S \cap X_1) = d$.
  \end{enumerate}
\end{lemma}

\begin{proof}
  Follows directly from the definition of a 2-join.
\end{proof}

\subsection{Stable sets overlapping 2-joins}
\label{sec:ineq}
We need kinds of blocks that preserve being in $\C{Berge}{}$.  To
define them we need several inequalities that tell more about how
stable sets and 2-joins overlap.

\begin{lemma}
  \label{l:ineqbasic}
  $0 \leq c \leq a, b \leq d \leq a+b$.
\end{lemma}

\begin{proof}
  The inequalities $0 \leq c \leq a, b \leq d$ are trivially true. Let
  $D$ be a maximum weighted stable set of $G[X_1]$.  We have:
  $$
  d = w(D) = w(D\cap A_1) + w(D\cap (C_1 \cup B_1)) \leq a + b.
  $$
\end{proof}

A 2-join with split $(X_1, X_2,A_1,B_1,A_2,B_2)$ is said to be
\emph{$X_1$-even} (resp.\ \emph{$X_1$-odd}) if all paths from $A_1$ to
$B_1$ with interior in $C_1$ are of even length (resp.\ odd length).
Note that from Lemma~\ref{l.2jAiBi}, if $G$ is in $\C{parity}{}$ and
$(X_1,X_2)$ is connected, then $(X_1, X_2)$ must be either $X_1$-even
or $X_1$-odd.

\begin{lemma}
  \label{l:ineqEven}
  If $(X_1, X_2)$ is an $X_1$-even 2-join of $G$, then $a+b \leq c+d$.
\end{lemma}

\begin{proof}
  Let $A$ be a stable set of $G[A_1 \cup C_1]$ of weight $a$ and $B$ a
  stable set of $G[B_1 \cup C_1]$ of weight $b$.  In the bipartite
  graph $G[A\cup B]$, we denote by $Y_A$ (resp.\ $Y_B$) the set of
  those vertices of $A\cup B$ such that there exists a path in $G[A
  \cup B]$ joining them to some vertex of $A \cap A_1$ (resp.\ $B \cap
  B_1$).  Note that from the definition, $A \cap A_1 \subseteq Y_A$,
  $B \cap B_1 \subseteq Y_B$ and no edges exist between $Y_A\cup Y_B$
  and $(A\cup B)\sm (Y_A \cup Y_B)$.  Also, $Y_A$ and $Y_B$ are
  disjoint with no edges between them because else, there is some path
  in $G[A\cup B]$ from some vertex of $A \cap A_1$ to some vertex of
  $B \cap B_1$.  If such a path is minimal with respect to this
  property, its interior is in $C_1$ and it is of odd length because
  $G[A\cup B]$ is bipartite.  This contradicts the assumption that
  $(X_1, X_2)$ is $X_1$-even.  Now we put:

  \begin{itemize}
  \item $Z_D = (A \cap Y_A) \cup (B \cap Y_B) \cup (A \sm (Y_A \cup Y_B))$;
  \item $Z_C = (A \cap Y_B) \cup (B \cap Y_A) \cup (B \sm (Y_A \cup Y_B))$.
  \end{itemize}

  From all the definitions and properties above, $Z_D$ and $Z_C$ are
  stable sets and $Z_D \subseteq X_1$ and $Z_C \subseteq C_1$.  So,
  $a+b = w(Z_C) + w(Z_D) \leq c+d$.
\end{proof}

\begin{lemma}
  \label{l:ineqOdd}
  If $(X_1, X_2)$ is an $X_1$-odd 2-join of $G$, then $c+d \leq a+b$.
\end{lemma}

\begin{proof}
  Let $D$ be a stable set of $G[X_1]$ of weight $d$ and $C$ a stable
  set of $G[C_1]$ of weight $c$.  In the bipartite graph $G[C\cup D]$,
  we denote by $Y_A$ (resp.\ $Y_B$) the set of those vertices of
  $C\cup D$ such that there exists a path in $G[C \cup D]$ joining
  them to some vertex of $D\cap A_1$ (resp.\ $D \cap B_1$).  Note that
  from the definition, $D \cap A_1 \subseteq Y_A$, $D \cap B_1
  \subseteq Y_B$ and no edges exist between $Y_A \cup Y_B$ and $(C\cup
  D)\sm (Y_A \cup Y_B)$.  Also, $Y_A$ and $Y_B$ are disjoint with no
  edges between them because else, there is some path in $G[C\cup D]$
  from some vertex of $D \cap A_1$ to some vertex of $D \cap B_1$.  If
  such a path is minimal with respect to this property, its interior
  is in $C_1$ and it is of even length because $G[C\cup D]$ is
  bipartite.  This contradicts the assumption that $(X_1, X_2)$ is
  $X_1$-odd.  Now we put:

  \begin{itemize}
  \item $Z_A = (D \cap Y_A) \cup (C \cap Y_B) \cup (C \sm (Y_A \cup
    Y_B))$;
  \item $Z_B = (D \cap Y_B) \cup (C \cap Y_A) \cup (D \sm (Y_A \cup Y_B)$.
  \end{itemize}

  From all the definitions and properties above, $Z_A$ and $Z_B$ are
  stable sets and $Z_A \subseteq A_1 \cup C_1$ and $Z_B \subseteq B_1
  \cup C_1$.  So, $c+d = w(Z_A) + w(Z_B) \leq a+b$.     
\end{proof}

\subsection{Even and odd blocks}

We call \emph{flat claw} of a weighted graph $G$ any set $\{q_1, q_2,
q_3, q_4\}$ of vertices such that:
\begin{itemize}
\item the only edges between the $q_i$'s are $q_1q_2$, $q_2q_3$ and
  $q_4q_2$;
\item $q_1$ and $q_3$ have no common neighbor in $V(G) \sm \{q_2\}$;
\item $q_4$ has degree~1 in $G$ and $q_2$ has degree~3 in $G$.
\end{itemize}

\begin{lemma}
  \label{l:4cClaw}
  Let $G$ be a graph, $Q = \{q_1, q_2, q_3, q_4\}$ a flat claw of $G$
  and $S'$ a maximum weighted stable set of $G$.  Then one and only
  one of the following holds:

    \begin{enumerate}
    \item\label{i:e1} $q_1\in S'$, $q_3\notin S'$ and $S' \cap Q$ is a
      maximum weighted stable set of $G[\{q_1, q_2, q_4\}]$;
    \item\label{i:e2} $q_1 \notin S'$, $q_3\in S'$ and $S'\cap Q$ is a
      maximum weighted stable set of $G[\{q_2, q_3, q_4\}]$;
    \item\label{i:e3} $q_1 \notin S'$, $q_3\notin S'$ and $S'\cap Q$
      is a maximum weighted stable set of $G[\{q_2, q_4\}]$;
    \item\label{i:e4}$q_1 \in S'$, $q_3\in S'$ and $S'\cap Q$ is a
      maximum weighted stable set of $G[\{q_1, q_2, q_3, q_4\}]$.
    \end{enumerate}
\end{lemma}

\begin{proof}
  Follows directly from the definitions.
\end{proof}

We define now the \emph{even block} $G_2$ with respect to
$(X_1, X_2)$.  We keep $X_2$ and replace $X_1$ by a flat claw on $q_1,
\dots, q_4$ where $q_1$ is complete to $A_2$ and $q_3$ is complete to
$B_2$.  We give the following weights: $w(q_1) = d-b$, $w(q_2) = c$,
$w(q_3) = d-a$, $w(q_4) = a+b-d$.  From Lemma~\ref{l:ineqbasic}, all
weights are non-negative.  By Lemma~\ref{l:ineqEven}, the following
Lemma applies in particular if $(X_1, X_2)$ is a connected $X_1$-even
2-join.

\begin{lemma}
  \label{l:evenBlock}
  If $a+b \leq c+d$ and if $G_2$ is the even block of $G$, then
  $\alpha(G_2) = \alpha (G)$.
\end{lemma}

\begin{proof}
  Let $S$ be a stable set of maximum weight in $G$.  Then $S$ must
  satisfy one of \ref{i:4c1}, \ref{i:4c2}, \ref{i:4c3} or \ref{i:4c4}
  of Lemma~\ref{l:4cases}.  Respective to these cases one can
  construct a stable set $S'$ of $G_2$ that has the weight of $S$, by
  taking the union of $S \cap X_2$ and one of $\{q_1, q_4\}$, $\{q_3,
  q_4\}$, $\{q_2\}$ or $\{q_1, q_3, q_4\}$.    

  Conversely, if $S'$ is a stable set of $G_2$ of maximum weight then
  it satisfies one of \ref{i:e1}, \ref{i:e2}, \ref{i:e3} or \ref{i:e4}
  of Lemma~\ref{l:4cClaw}.  Respective to these cases, $w(S'\cap Q)$
  is $a$, $b$, $c$ or $d$ (by Lemma~\ref{l:ineqbasic} and because $a+b
  \leq c+d$) and one can construct a maximum stable set $S$ of $G$ by
  replacing $S' \cap Q$ by a maximum weighted stable set of $G[A_1
  \cup C_1]$, $G[B_1 \cup C_1]$, $G[C_1]$ or $G[X_1]$.
\end{proof}

We call \emph{flat vault} of graph $G$ any set $\{r_1, r_2,
r_3, r_4, r_5, r_6\}$ of vertices such that:

\begin{itemize}
\item the only edges between the $r_i$'s are such that
  $r_3,r_4,r_5,r_6,r_3$ is a 4-hole;
\item $N(r_1) = N(r_5)\sm \{r_4, r_6\}$; 
\item $N(r_2) = N(r_6)\sm \{r_3, r_5\}$; 
\item $r_1$ and $r_2$ have no common neighbors;
\item $r_3$ and $r_4$ have degree 2 in $G$.
\end{itemize}

\begin{lemma}
  \label{l:4cVault}
  Let $G$ be a graph, $Q = \{r_1, r_2, r_3, r_4, r_5, r_6\}$ a flat vault of $G$
  and $S'$ a maximum weighted stable set of $G$.  Then one and only
  one of the following holds:
 
  \begin{enumerate}
  \item\label{i:o1} $S' \cap \{r_1, r_5\} \neq \emptyset$, $S' \cap
    \{r_2, r_6\} = \emptyset$ and $S' \cap Q$ is a maximum weighted
    stable set of $G[\{r_1, r_3, r_4, r_5\}]$;
  \item\label{i:o2} $S' \cap \{r_1, r_5\} = \emptyset$, $S' \cap
    \{r_2, r_6\} \neq \emptyset$ and $S'\cap Q$ is a maximum weighted
    stable set of $G[\{r_2, r_3, r_4, r_6\}]$;
  \item\label{i:o3} $S' \cap \{r_1, r_5\} = \emptyset$, $S' \cap
    \{r_2, r_6\} = \emptyset$ and $S'\cap Q$ is a maximum weighted
    stable set of $G[\{r_3, r_4\}]$;
  \item\label{i:o4} $S' \cap \{r_1, r_5\} \neq \emptyset$, $S' \cap
    \{r_2, r_6\} \neq \emptyset$ and $S'\cap Q$ is a maximum weighted
    stable set of $G[\{r_1, r_2, r_3, r_4, r_5, r_6\}]$.
  \end{enumerate}
\end{lemma}

\begin{proof}
  Follows directly from the definitions. 
\end{proof}

Let us now define the \emph{odd block} $G_2$ with respect to 
$(X_1,X_2)$.  We replace $X_1$ by a flat vault on $r_1, \dots, r_6$.
Moreover $r_1, r_5$ are complete to $A_2$ and $r_2, r_6$ are complete
to $B_2$.  We give the following weights: $w(r_1) = d-b$, $w(r_2) =
d-a$, $w(r_3) = w(r_4) = c$, $w(r_5) = w(r_6) = a+b-c-d$.  Note that
if we suppose $c+d \leq a+b$ (which holds in particular if $(X_1,
X_2)$ is an $X_1$-odd connected 2-join by Lemma~\ref{l:ineqOdd}), all
the weights are non-negative by Lemma~\ref{l:ineqbasic}.

\begin{lemma}
  \label{l:oddBlock}
  If $c+d \leq a+b$ and if $G_2$ is the odd block of $G$, then
  $\alpha(G_2) = \alpha (G)$.
 \end{lemma}

\begin{proof}
  Let $S$ be a stable set of maximum weight in $G$.  Then $S$ must
  satisfy one of \ref{i:4c1}, \ref{i:4c2}, \ref{i:4c3} or \ref{i:4c4}
  of Lemma~\ref{l:4cases}.  So, respective to these cases, it is easy
  to construct a stable set $S'$ of $G_2$ that has the weight of $S$,
  by taking the union of $S \cap X_2$ and one of $\{r_1, r_3, r_5\}$,
  $\{r_2, r_4, r_6\}$, $\{r_3\}$ or $\{r_1, r_2, r_3, r_5\}$.

  Conversely, if $S'$ is a stable set of $G_2$ of maximum weight then
  it satisfies one of \ref{i:o1}, \ref{i:o2}, \ref{i:o3} or \ref{i:o4}
  of Lemma~\ref{l:4cVault}.  Respective to these cases, $w(S'\cap Q)$
  is $a$, $b$, $c$ or $d$ (because $c+d \leq a+b$) and one can
  construct a maximum weighted stable set $S$ of $G$ of the same
  weight as $S'$ by replacing $S' \cap \{r_1, r_2, r_3, r_4, r_5,
  r_6\}$ by a maximum weighted stable set of $G[A_1 \cup C_1]$, $G[B_1
  \cup C_1]$, $G[C_1]$ or $G[X_1]$.
\end{proof}

Note that the following lemma fails for $\C{ehf}{}$ because the odd
block contains an even hole.
 
\begin{lemma}
  \label{l:stayBerge}
  Let $G$ be a graph in $\C{Berge}{}$ and $(X_1, X_2)$ be a connected
  2-join of $G$.   If $(X_1, X_2)$ is $X_1$-even then the even block
  $G_2$ is in $\C{Berge}{}$. If $(X_1, X_2)$ is $X_1$-odd then the
  odd block $G_2$ is in $\C{Berge}{}$.
\end{lemma}

\begin{proof}
  Suppose that $G_2$ contains an odd hole $H$.  If no edge of $H$ has
  both ends in $V(G_2)\sm X_2$, then $H \subseteq X_2 \cup
  (N_{G_2}(A_2) \sm X_2) \cup (N_{G_2}(B_2) \sm X_2)$.  We obtain an
  odd hole $H'$ of $G$ as follows.  By Lemma~\ref{k1} \ref{i:BiComp1},
  there exist non-adjacent vertices $a_1\in A_1$, $b_1 \in B_1$.  If
  $H \cap (N_{G_2}(A_2) \sm X_2) \neq \emptyset$, we replace the
  unique vertex in $H \cap (N_{G_2}(A_2) \sm X_2)$ by $a_1$.  We
  proceed similarly with $H \cap (N_{G_2}(B_2) \sm X_2)$ and $b_1$.
  We obtain an odd hole $H'$ of $G$, a contradiction. 

  If $H$ has an edge whose ends are both in $V(G_2)\sm X_2$ then $H$
  is vertex-wise partitioned into $q_1 \tp q_2 \tp q_3$ when $(X_1,
  X_2)$ is $X_1$-even (resp.\ $r_5 \tp r_6$ when $(X_1, X_2)$ is
  $X_1$-odd), and a path with one end in $A_2$, one end in $B_2$ and
  interior in $C_2$.  Then an odd hole of $G$ can be obtained by
  replacing $q_1 \tp q_2 \tp q_3$ (resp.\ $r_5 \tp r_6$) by a path of
  even (resp.\ odd) length of $G$ from $A_1$ to $B_1$ with interior in
  $C_1$.  This contradicts $G$ being Berge.

  Suppose that $G_2$ contains an odd antihole $H$.  Since an antihole
  on 5 vertices is in fact a hole, we may assume that $H$ is on at
  least 7 vertices.  So all vertices of $H$ have degree at least four.
  Hence, if $G_2$ is an even block then $H$ cannot go through $q_2,
  q_4$. So, up to the replacement of at most two vertices, $H$ is an
  odd antihole of $G$, a contradiction.  Now suppose $G_2$ is an odd
  block.  Because of the degrees, $r_3, r_4 \notin H$.  In an antihole
  on at least 7 vertices, every pair of vertices has a common
  neighbor.  A vertex of $\{r_1, r_5\}$ has no common neighbor with a
  vertex of $\{r_2, r_6\}$.  So, we may assume that $H\cap \{r_2,
  r_6\} = \emptyset$.  We have $N_{G_2}(r_1) \subseteq N_{G_2}(r_5)$
  so not both $r_1, r_5$ are in $H$.  So, we may assume that
  $r_5\notin H$.  So, up to the replacement of $r_1$ by a vertex of
  $A_1$, $H$ is an odd antihole of $G$, a contradiction.
\end{proof}

\subsection{The gem block}
\label{subsec:gemblock}

We present here a block of decomposition that we do not use in the
rest of the paper but that is interesting because it can be used in
all situations (whereas some inequalities must be satisfied for even
and odd blocks).  

To build the \emph{gem-block} $G_2$ replace $X_1$ by an induced path
$p \tp x \tp y \tp p'$ plus a vertex $z$ complete to this path.
Vertex $p$ is complete to $A_2$ and vertex $p'$ is complete to $B_2$.
We give weights: $w(p) = a$, $w(x) = a+b-d$, $w(y) = d$, $w(p') = 2d -
a$, $w(z) = c+d$.  Note that all weights are non-negative by
Lemma~\ref{l:ineqbasic}.  We omit the proof of the following Lemma
since we do not use it.

\begin{lemma}\label{lineblock}
  If $G_2$ is the gem-block of $G$ then $\alpha(G_2) = \alpha(G) + d$.
\end{lemma}

The gem-block appears implicitly in the proof of the NP-completeness
result in Section~\ref{sec:npc}.

\section{Extensions of basic classes}
\label{sec:extension}

To build a decomposition tree that allows keeping track of maximum
stable sets we use the even and odd blocks defined in
Section~\ref{sec:alphaTrack}.  As a consequence, the leaves of our
decomposition tree may fail to be basic, but are what we call
\emph{extensions} of basic graphs.  Let us define this.

Let $P = p_1 \tp \cdots \tp p_k$, $k\geq 4$, be a flat path of a graph
$G$.  \emph{Extending} $P$ means:

\begin{itemize}
\item Either:
  \begin{enumerate}
  \item replace the vertices of $P$ by a flat claw on $q_1, \dots, q_4$
    where $q_1$ is complete to $N_G(p_1)\sm \{p_2\}$ and $q_3$ is
    complete to $N_G(p_k) \sm \{p_{k-1}\}$;
  \item replace $X_1$ by a flat vault on $r_1, \dots, r_6$ where $r_1,
    r_5$ are complete to $N_G(p_1)\sm \{p_2\}$ and $r_2, r_6$ are complete
    to $N_G(p_k) \sm \{p_{k-1}\}$.
  \end{enumerate}
\item Mark the vertices of the flat claw (or vault) with the integer
  $k$.
\end{itemize}

An \emph{extension} of a pair $(G, {\cal M})$, where $G$ is a graph
and $\cal M$ is a set of vertex-disjoint flat paths of length at
least~3 of $G$, is any weighted graph obtained by extending the flat
paths of $\cal M$ and giving any non-negative weights to all the vertices.
Note that since $\cal M$ is a set of \emph{vertex-disjoint} paths, the
extensions of the paths from $\cal M$ can be done in any order and
lead to the same graph.  An \emph{extension} of a graph $G$ is any
graph that is an extension of $(G, {\cal M})$ for some $\cal M$.

We say that the extension of $P$ is \emph{parity-preserving} when $P$
has even length and is replaced by a flat claw, or when $P$ has odd
length and is replaced by a flat vault. We define the
\emph{parity-preserving} extension of a pair $(G, {\cal M})$ and of a
graph $G$ by requiring that all extensions of paths are
parity-preserving.

\subsection{Recognition of extensions basic graphs}

We will describe algorithms for computing cliques and stable sets in
graphs from our basic classes and their extensions.  To apply these
algorithms we need to detect in which basic class a graph is.  For
bipartite graphs, line graphs of bipartite graphs and their
complements, this a classical problem,
see~\cite{lehot:root,roussopoulos:linegraphe}.  For double split
graphs, this is very easy, see \cite{nicolas:bsp}, Section~7.  For
path-cobipartite graphs it is very easy by picking a vertex of
degree~2 if any, checking if it belongs to a flat path, if so taking a
maximal such flat path, and checking if the flat path satisfies the
definition of a path-cobipartite graph.  A similar trick recognizes
path-double split graphs.  All these classes can be recognized in
linear time.

The class $\C{ehf}{basic}$ can be recognized in time ${\cal
  O}(n^2m)$ by checking for all pairs of vertices if their deletion
gives the line graph of a tree.  Checking that the graph is even-hole
free is easy, since in the line graph of a tree, there exists a unique
induced path joining any pair of vertices. 

Recognition of extensions of basic graphs is easy thanks to the mark
given to the new vertices arising from extensions. These marks allow to
compute the original graph from its extension.

Also, when a graph is identified to be basic, the algorithms above
certify that the graph is basic.  For bipartite graphs, it gives a
bipartition, for a line graph $G$, a root-graph $R$ such that
$G=L(R)$.  For double split graphs, path cobipartite graphs,
path-double split graph and graph from $\C{ehf}{basic}$, the sets like
in their respective definitions are output.

We do not write a theorem about these algorithms, but in the
description of the algorithms in the rest of the paper, when we
consider an extended basic graph, it is implicit that the algorithm
can know in time $O(n^2m)$ in which basic class the graph is.  Since
all our algorithms run in time at least ${\cal O}(n^3m)$, this does not
affect the overall complexity.

\subsection{Parity-preserving extensions of basic Berge graphs (except
  line graphs)}

\begin{lemma}
  \label{l:staybip}
  A parity-preserving extension of a bipartite graph is a bipartite
  graph.
\end{lemma}

\begin{proof}
  Suppose that a graph $G$ is bipartite.  So we color its vertices
  black and white.  Suppose that a parity-preserving extension with
  respect to $P= p_1 \tp \cdots \tp p_k$ is performed.  If the path
  has even length then up to symmetry $p_1$ and $p_k$ are black.
  Since the extension is parity-preserving, $P$ is replaced by a flat
  claw on $q_1, q_2, q_3, q_4$.  We give color black to $q_1, q_3,
  q_4$ and color white to $q_2$.  If the path has odd length then up
  to symmetry $p_1$ is black and $p_k$ is white.  Since the extension
  is parity-preserving, $P$ is replaced by a flat vault on $r_1, r_2,
  r_3, r_4, r_5, r_6$.  We give color black to $r_1, r_3, r_5$ and
  color white to $r_2, r_4, r_6$.  This shows that the
  parity-preserving extension of $P$ yields a bipartite graphs and the
  lemma follows by an induction on the number of extended paths.
\end{proof}

The following lemma shows that maximum weighted stable sets can be
computed for all parity-preserving extensions of Berge basic classes,
except line graphs of bipartite graphs.

\begin{lemma}
  \label{l:optBergeBasic}
  There is an algorithm with the following specification:
  \begin{description}
  \item[Input: ] A weighted graph $G$ that is a parity preserving
    extension of either a bipartite graph, the complement of a
    bipartite graph, the complement of a line graph of a bipartite
    graph, a path-cobipartite graph, the complement of a
    path-cobipartite graph, a path-double split graph or the
    complement of a path-double split graph.
  \item[Output: ] A maximum weighted stable set of $G$.
  \item[Running time: ] ${\cal O} (n^5)$
  \end{description}
\end{lemma}

\begin{proof}
  For parity-preserving extensions of bipartite graphs, the result
  follows from Lemma~\ref{l:staybip}.  Indeed, computing a maximum
  weighted stable set in a bipartite graph can be done in time ${\cal
    O} (n^3)$, see~\cite{schrijver}.
  
  Let $k$ be a constant integer and $\cal C$ a class of graphs for
  which there exits a polynomial time algorithm to compute maximum
  weighted stable sets.  Let ${\cal C}_k$ be the class of those graphs
  obtained from a graph in $\cal C$ by adding $k$ vertices and giving
  a mark to them.  Then there is a polynomial time algorithm for
  computing a maximum weighted stable set for a graph $G$ in ${\cal
    C}_k$.  It suffices to try every stable subset $S$  of the set
  of marked vertices, to delete all the marked vertices, to give
  weight zero to the neighbors of vertices of $S$, to run the
  algorithm for $\cal C$ in what remains and to denote by $A_S$ the
  stable set obtained.  Then compute $w(S\cup A_S)$.  Choose a stable
  set of maximum weight among the $N \leq 2^k$ stable sets so
  obtained.  Note that $2^k$ is a constant.

  This method works for parity-preserving extensions of complements of
  bipartite graphs, complements of line graphs of bipartite graphs and
  complements of path-cobipartite graphs.  Indeed, as we
  show next, a graph from any of these classes cannot contain two
  vertex-disjoint flat paths of length at least three.  So, at most one
  path is extended and by the remark above the desired algorithm
  relies on classical algorithms for maximum weighted stable set in
  complements of bipartite graphs, complements of line graphs of
  bipartite graphs, and bipartite graphs (note that maximum weighted
  stable set in a complement of a path-cobipartite graph corresponds
  to maximum weighted clique in path-cobipartite graph, and all
  maximal cliques of such graphs are either of size~2 or belong to the
  cobipartite graph obtained by removing vertices of degree~2). All this
  can be done in time ${\cal O} (n^3)$, see \cite{schrijver}.

  The only complement of a bipartite graph that contains a flat path
  of length at least~3 is $P_4$.  line graphs of bipartite graphs
  cannot contain diamonds, so complements of line graphs of bipartite
  graphs cannot contain complements of diamonds.  Hence they cannot
  contain two vertex-disjoint flat paths of length at least three.
  Now we deal with the complement $G$ of a path-cobipartite graph.
  Since we know how to handle complements of bipartite graphs, we may
  assume that the path $P$ form the definition is non-empty.  So, $G$
  contains a vertex $u$ of degree $|V(G)|-3$ (pick a vertex in $P$).
  So, $G$ cannot contain two disjoint flat paths of length at least~3
  because the interior vertices of such paths would contradict
  $\deg(u) = |V(G)|-3$.

  To compute a maximum weighted stable set in a parity-preserving
  extension $G$ of a path-cobipartite graphs $H$, apply the following
  method, where the notation $A, B, P$ is like in the definition of
  path-cobipartite graphs.  First observe that only vertices of $P$
  are replaced during extensions.  For all stable sets $S$ of $G[A
  \cup B]$ (and there are at most $|A|+|B|+|A||B|$ of them), consider
  the graph $G_S=G\setminus (A \cup B \cup N(S))$. Note that $G_S$ is
  a bipartite graph because it is an induced subgraph of a
  parity-preserving extension of a path.  So we can compute a maximum
  weighted stable set $T_S$ of $G_S$ in time ${\cal O} (n^3)$, see
  \cite{schrijver}.  Among all stable sets $S \cup T_S$ so constructed,
  choose one of maximum weight.  So all this can be done in time
  ${\cal O} (n^5)$.

  To compute a maximum weighted stable set in a parity-preserving
  extension $G$ of a path-double split graph $H$ apply the following
  method, where the notation $A, B, C, D, E, k, l$ is like in the
  definition.  First observe that only vertices of $E$ are replaced
  during the extension.  For all stable sets $S$ of $G[C \cup D]$ (and
  there are $3l+1$ of them, including $\emptyset$), consider the graph
  $G_S=G\sm (C \cup D \cup N(S))$.  So, $G_S$ is an induced subgraph
  of $G \sm (C \cup D)$ and has at most $k$ connected components that
  are paths or parity-preserving extensions of paths, and hence $G_S$
  is a bipartite graph.  So we can compute a maximum weighted stable
  set $T_S$ of $G_S$ in time ${\cal O} (n^3)$, see \cite{schrijver}.
  Among all stable sets $S \cup T_S$ so constructed, choose one of
  maximum weight.  So all this can be done in time ${\cal O} (n^4)$.

  To compute a maximum weighted stable set in a parity-preserving
  extension $G$ of the complement $\overline{H}$ of a path-double
  split graph $H$ there are two cases.  First case, the set $E$ is
  empty.  Then, $H$ is in fact a double split graph and so is
  $\overline{H}$.  So $G$ is a parity-preserving extension of a
  path-double split graph, and we already know how to proceed in this
  case.  Second case, the set $E$ is not empty.  Then, all vertices in
  $H$ have at least 3 non-neighbors, so in $\overline{H}$, no vertex
  has degree~2.  So, no path can be extended, $G=\overline{H}$ and to
  compute a maximum weighted stable set in $G$ it suffices to compute
  a maximum weighted clique in $H$.  We can do this by listing all
  cliques $K$ of $H[A \cup B \cup E]$ (including the empty set).  Note
  that there are only linearly many such cliques.  Let $H_K$ be the
  subgraph of $H$ induced by the set of all vertices $C \cup D$ that are
  adjacent to all of $K$. It is easy to compute a maximum weighted
  clique $T_K$ of $H_K$ (it suffices to choose for each pair
  $c_j,d_j$, $j=1, \ldots ,m$, the vertex with bigger weight). Among
  all cliques $K \cup T_K$ so constructed, choose the one of maximum
  weight.  So all this can be done in time ${\cal O} (n^4)$.
\end{proof}

\subsection{Extensions of line graphs}

Extensions of line graphs are more difficult to handle than other
extensions because an extension of a line graph may fail to be a
line graph and a line graph may contain arbitrarily many disjoint
long flat paths.  Note that in this subsection, extensions are \emph{not
required to be parity-preserving}.

Let $G'$ be a weighted graph that is an extension of a line graph
$G=L(R)$, see Figures~\ref{fig:RG} and \ref{fig:Gp}.  We now define
the \emph{transformation} $G''$ of $G'$, see Figure~\ref{fig:GppRpp}.
The structure of $G''$, i.e.\ its vertices and edges, depends only on
$G$ but the weights given to its vertices depend only on $G'$.  Let
$\cal M$ be the set of vertex-disjoint flat paths of length at least~3
of $G$ that are extended to get~$G'$.  So, ${\cal M} = \{P^1, \dots,
P^k\}$ and we put $P^i = p^i_1 \tp \cdots \tp p^i_{l_i}$.  For all
$1\leq i \leq k$, path $P^i$ of $G$ is replaced in $G'$ by a set $Q^i$ that
induces either a flat claw on vertices $q^i_1, q^i_2, q^i_3, q^i_4$ or
a flat vault on vertices $r^i_1, r^i_2, r^i_3, r^i_4, r^i_5, r^i_6$.
For all flat paths $P^i$ of $\cal M$, we put $A^i_2= N_{G}(p^i_1)\sm
\{p^i_2\}$, $B^i_2 = N_{G}(p^i_{l_i}) \sm \{p^i_{l_i - 1}\}$.

\begin{figure}
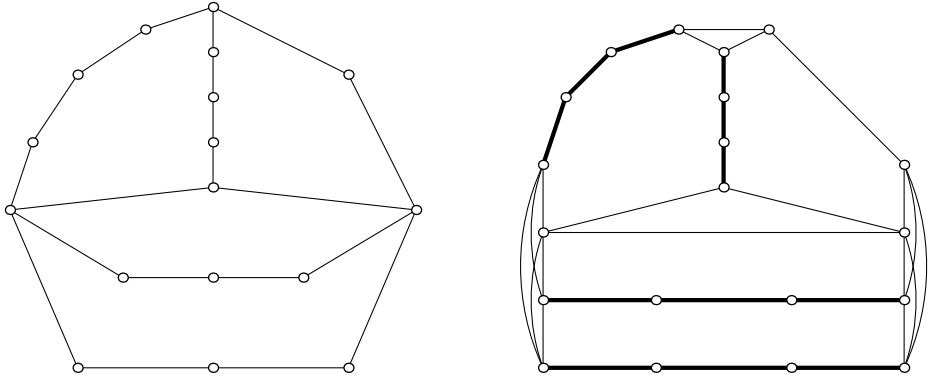

  \begin{center}
    \includegraphics{lineExt.1}\hspace{3em}
    \includegraphics{lineExt.2}
   \caption{$R$ and $L(R)$\label{fig:RG}. Paths of $L(R)$ to be
     extended are represented with bold edges.}
  \end{center}
\end{figure}

\begin{figure}
  \begin{center}
    \includegraphics{lineExt.3}
   \caption{$G'$\label{fig:Gp}}
  \end{center}
\end{figure}

\begin{figure}
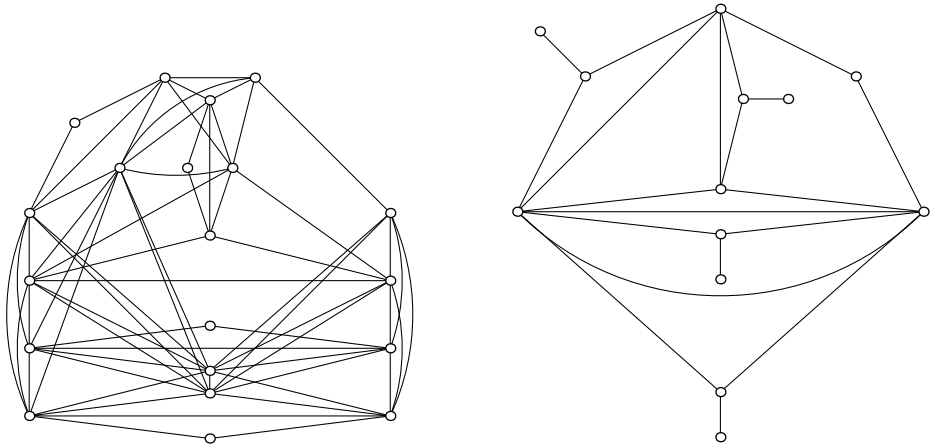

  \begin{center}
    \includegraphics{lineExt.4}\hspace{3em}
    \includegraphics{lineExt.5}
   \caption{$G''=L(R'')$ and $R''$\label{fig:GppRpp}}
  \end{center}
\end{figure}

For all $1\leq i \leq k$, we prepare a set $S^i$ of four new vertices
$p^i, p'^i, x^i, y^i$.  The graph $G''$ has vertex-set:
$$V(G'') = (S^1 \cup \cdots \cup S^k ) \cup V(G) \sm (P^1 \cup \cdots \cup P^k).$$

\noindent Edges of $G''$ depend only on edges of $G$.  They are:

\begin{itemize}
\item $p^ip'^i$, $x^ip^i$, $p^iy^i$, $y^ip'^i$, $p'^ix^i$, $i= 1,
  \dots, k$;
\item $uv$ for all $u, v \in V(G) \cap V(G'')$ such that $uv\in
  E(G)$;
\item $p^iu$ for all $u\in A^i_2 \cap V(G'')$, $i=1, \dots, k$;
\item $p'^iu$ for all $u\in B^i_2 \cap V(G'')$, $i=1, \dots, k$;
\item $x^iu$ for all $u\in (A^i_2 \cup B^i_2) \cap V(G'')$, $i=1,
  \dots, k$;
\item $p^ip^j$ for all $i, j$ such that $p^i_1p^j_1 \in E(G)$;
\item $p'^ip^j$ for all $i, j$ such that $p^i_{l_i}p^j_1 \in E(G)$;
\item $p'^ip'^j$ for all $i, j$ such that $p^i_{l_i}p^j_{l_j} \in
  E(G)$;
\item $x^ip^j$ for all $i, j$ such that $p^i_1p^j_1 \in E(G)$ or
  $p^i_{l_i}p^j_1 \in E(G)$;
\item $x^ip'^j$ for all $i, j$ such that $p^i_1p^j_{l_j} \in E(G)$ or
  $p^i_{l_i}p^j_{l_j} \in E(G)$;
\item $x^ix^j$ for all $i, j$ such that $p^i_1p^j_{1} \in E(G)$ or
  $p^i_{1}p^j_{l_j} \in E(G)$ or $p^i_{l_i}p^j_{1} \in E(G)$ or
  $p^i_{l_i}p^j_{l_j} \in E(G)$.
\end{itemize}

\noindent We define the following numbers that depend only on $G'$:

\begin{itemize}
\item $a^i = \alpha(G'[\{q^i_1, q^i_2, q^i_4\}])$ for all $i$ such
  that $Q^i$ is a flat claw of $G'$;
\item $a^i = \alpha(G'[\{r^i_1, r^i_3, r^i_4, r^i_5\}])$ for all
  $i$ such that $Q^i$ is a flat vault of $G'$;
\item $b^i = \alpha(G'[\{q^i_2, q^i_3, q^i_4\}])$ for all $i$ such
  that $Q^i$ is a flat claw of $G'$;
\item $b^i = \alpha(G'[\{r^i_2, r^i_3, r^i_4, r^i_6\}])$ for all
  $i$ such that $Q^i$ is a flat vault of $G'$;
\item $c^i = \alpha(G'[\{q^i_2, q^i_4\}])$ for all $i$ such that
  $Q^i$ is a flat claw of $G'$;
\item $c^i = \alpha(G'[\{r^i_3, r^i_4\}])$ for all $i$ such that
  $Q^i$ is a flat vault of $G'$;
\item $d^i = \alpha(G'[\{q^i_1, q^i_2, q^i_3, q^i_4\}])$ for all
  $i$ such that $Q^i$ is a flat claw of $G'$;
\item $d^i = \alpha(G'[\{r^i_1, r^i_2, r^i_3, r^i_4, r^i_5,
  r^i_6\}])$ for all $i$ such that $Q^i$ is a flat vault of $G'$.
\end{itemize}

\noindent Note that from the definitions, $c^i \leq a^i, b^i \leq d^i$
for all $i=1, \dots, k$.  We give the following weights to the
vertices of $G''$ (they depend on the weights in $G'$):

\begin{itemize}
\item $w_{G''}(u) = w_{G'}(u)$ for all $u\in V(G) \cap V(G'')$;
\item $w_{G''}(p^i) = a^i$, $i= 1, \dots, k$;
\item $w_{G''}(p'^i) =b^i$, $i= 1, \dots, k$;
\item $w_{G''}(y^i) = c^i$, $i= 1, \dots, k$;
\item $w_{G''}(x^i) = d^i-c^i$, $i= 1, \dots, k$.
\end{itemize}

A \emph{multigraph} is a graph where multiple edges between vertices
are allowed (but we do not allow loops). 

\begin{lemma}
  \label{l:L(multi)}
  $G''$ is the line graph of a multigraph.
\end{lemma}

\begin{proof}
  Path $P^i$ of $G$ corresponds to a path $R^i = r^i_1 \tp \cdots \tp
  r^i_{l_i+1}$ of $R$.  For all $i=1, \dots, k$, path $R^i$ is induced
  and has interior vertices of degree 2 in $R$ because $P^i$ is a flat
  path of $G$.  Since paths of $\cal M$ are vertex-disjoints, paths
  $R^1, \dots, R^k$ are edge-disjoint (but they may share
  end-vertices).  Now let us build a multigraph $R''$ from $R$, see
  Figure~\ref{fig:GppRpp}.  We delete the interior vertices of all
  $R^i$'s.  For each $R^i$, we add two vertices $u^i$, $v^i$ and the
  edges $r^i_1r^i_{l_i+1}$, $u^ir^i_1$, $u^ir^i_{l_i+1}$ and $u^iv^i$.

  It is a routine matter to check that $L(R'')$ is isomorphic to
  $G''$.  Edge $r^i_1r^i_{l_i+1}$ corresponds to vertex $x^i$, edge
  $u^ir^i_1$ corresponds to vertex $p^i$, edge $u^ir^i_{l_i+1}$
  corresponds to $p'^i$ and edge $u^iv^i$ corresponds to vertex $y^i$.
  Note that possibly, two paths $R^i$ and $R^j$ have the same ends.
  For instance $r^i_1=r^j_1$ and $r^i_{l_i+1} = r^j_{l_j+1}$ is
  possible.  Then, the edge $r^i_1r^i_{l_i+1}$ is added twice.  This
  is why we need $R''$ to be a multigraph.
\end{proof}

\begin{lemma}
  \label{l:Linealpha}
  $\alpha(G'') = \alpha(G')$.
\end{lemma}

\begin{proof}
  Let $S'$ be a stable set of maximum weight in $G'$.  Let us build a
  stable set $S''$ of $G''$ of same weight.  In $S''$, we keep all
  vertices of $S' \cap V(G'')$.  For all~$i$ such that $Q^i$ is a flat
  claw of $G'$, $S'$ satisfies one of \ref{i:e1}, \ref{i:e2},
  \ref{i:e3} or \ref{i:e4} of Lemma~\ref{l:4cClaw}.  So, respective to
  these cases we put one of $\{p^i\}$, $\{p'^i\}$, $\{y^i\}$ or
  $\{x^i, y^i\}$ in $S''$.  For all~$i$ such that $Q^i$ is a flat
  vault of $G'$, $S'$ satisfies one of \ref{i:o1}, \ref{i:o2},
  \ref{i:o3} or \ref{i:o4} of Lemma~\ref{l:4cVault}.  So, respective to
  these cases we put one of $\{p^i\}$, $\{p'^i\}$, $\{y^i\}$ or
  $\{x^i, y^i\}$ in $S''$.  This yields a stable set of $G''$ of the same
  weight as $S'$. 

  Conversely, if $S''$ is a stable set of $G''$ of maximum weight then
  we may assume for all~$i$, $S''\cap \{p^i, p'^i, x^i, y^i\}$ is one
  of $\{p^i\}$, $\{p'^i\}$, $\{y^i\}$ or $\{x^i, y^i\}$.  The only
  exception could be when $w_{G''}(y^i) = 0$ and $S''\cap \{p^i, p'^i,
  x^i, y^i\} = \emptyset$ or $\{x^i\}$, but then we add $y^i$ to
  $S''$.  If $Q^i$ is a flat claw, respective to these cases, we put
  in $S'$ a maximum weighted stable set of one of $G'[\{q^i_1, q^i_2,
  q^i_4\}]$, $G'[\{q^i_2, q^i_3, q^i_4\}]$, $G'[\{q^i_2, q^i_4\}]$ or
  $G'[\{q^i_1, q^i_2, q^i_3, q^i_4\}]$.  If $Q_i$ is a flat vault,
  respective to these cases, we put in $S'$ a maximum weighted stable
  set of one of $G'[\{r^i_1, r^i_3, r^i_4, r^i_5\}]$, $G'[\{r^i_2,
  r^i_3, r^i_4, r^i_6\}]$, $G'[\{r^i_3, r^i_4\}]$ or $G'[\{r^i_1,
  r^i_2, r^i_3, r^i_4, r^i_5, r^i_6\}]$.  This yields a stable set of
  $G'$ of the same weight as~$S''$.
\end{proof}

\begin{lemma}
  \label{l:optLineGraphs}
  There is an algorithm with the following specification:
  \begin{description}
  \item[Input: ] A weighted graph $G'$ that is an extension of a
    line graph $G$. 
  \item[Output: ] A maximum weighted stable set of $G'$.
  \item[Running time: ] ${\cal O} (n^3)$
  \end{description}
\end{lemma}

\begin{proof}
  Build the transformation $G''$ of $G'$ as explained above.  So, by
  Lemma~\ref{l:L(multi)}, $G''$ is the line graph of a multigraph.
  Compute a multigraph $R$ such that $G''=L(R)$
  (see~\cite{lehot:root,roussopoulos:linegraphe}), then compute in $R$
  a matching of maximum weight by Edmonds' algorithm
  (see~\cite{schrijver,{edmonds:ptf}}).  It corresponds to a maximum
  weighted stable set in $G''$.  By Lemma~\ref{l:Linealpha}, this
  maximum weighted stable set has the same weight as a maximum weighted
  stable set $S'$ of $G'$.  Note that the proof of
  Lemma~\ref{l:Linealpha} shows how to actually obtain $S'$.
\end{proof}

\subsection{Extensions of basic even-hole-free graphs}

Here again, extensions are \emph{not required to be
  parity-preserving}.

\begin{lemma}
  \label{l:optEhfBasic}
  There is an algorithm with the following specification:
  \begin{description}
  \item[Input: ] A weighted graph $G$ that is an extension of a graph
    from $\C{ehf}{basic}$.
  \item[Output: ] A maximum weighted stable set and a maximum weighted
    clique of $G$.
  \item[Running time: ] ${\cal O} (n^4)$
  \end{description}
\end{lemma}

\begin{proof}
  Let $H$ be a graph in $\C{ehf}{basic}$ and ${\cal M}$ a set of
  vertex-disjoint flat paths of $H$ of length at least 3, such that
  the input graph $G$ is an extension of $(H, {\cal M})$.  Let $A$ be
  a set of vertices of $H$ such that $|A| \leq 2$ and $H \setminus A$
  is a line graph (it takes time ${\cal O}(n^4)$ to find $A$, by checking for
  all possible pairs whether their removal yields a line graph).  Let
  ${\cal M}'$ be the set of paths of ${\cal M}$ that contain some
  vertex of $A$.  Note that $|{\cal M}'| \leq 2$.

  Since $|A| \leq 2$ and $|{\cal M}'| \leq 2$, there is a set of
  vertices $B$ such that $|B| \leq 12$ and $G$ is obtained from $G'$
  by adding vertices of $B$ where $G'$ is an extension of $(H
  \setminus A, {\cal M}\setminus {\cal M}')$.  By Lemma
  \ref{l:optLineGraphs}, a maximum weighted stable set of $G'$ can be
  computed in time ${\cal O} (n^3)$.  To compute a maximum weighted
  stable set of $G$ it suffices to try every stable set $S$ of $B$, to
  delete all vertices of $S\sm B$, to give weight zero to neighbors of
  vertices of $S$, and to compute the maximum weighted stable set of
  the remaining graph (which is as we noted an extension of a
  line graph).  Choose a maximum weighted stable set so
  obtained. Clearly all this can be done in time ${\cal O} (n^3)$.

  To compute a maximum weighted clique, it suffices to notice that $G$
  contains linearly many inclusion-wise maximal cliques.  So, it
  suffices to list them and to choose one of maximum weight. 
\end{proof}

It is convenient to sum up all the results of the section:

\begin{lemma}
  \label{l:optBasic}
  There is an algorithm with the following specification:
  \begin{description}
  \item[Input: ] A weighted graph $G$ that is a parity preserving
    extension of a graph from $\C{Berge}{basic}$ or any extension of a
    graph from $\C{ehf}{basic}$.
  \item[Output: ] A maximum weighted stable set and a maximum weighted
    clique of $G$.
  \item[Running time: ] ${\cal O} (n^5)$
  \end{description}
\end{lemma}

\begin{proof} 
  Follows from Lemmas~\ref{l:optBergeBasic}, \ref{l:optLineGraphs}
  and~\ref{l:optEhfBasic} for computing stables sets.  For cliques, it
  is done in Lemma~\ref{l:optEhfBasic} for graphs from
  $\C{ehf}{basic}$.  For $\C{Berge}{basic}$, it suffices to notice that
  the class is self-complementary, so we may rely on the algorithm for
  stable sets. 
\end{proof}

\section{Constructing the decomposition tree}
\label{sec:tree}

We now give algorithms to construct several decomposition trees for
graphs in our classes.  First we show how to build a decomposition
tree with the usual parity preserving blocks (as defined in
Section~\ref{sec:blocks}).  Then we show how to reprocess such a tree
to get a tree with clique-blocks, even blocks or odd blocks according
to what we need to optimize.

\subsection{Tree with parity preserving blocks}

We define now a decomposition tree $T_G$ of a graph $G \in {\cal D}$,
where ${\cal D}$ is one of $\C{Berge}{no cutset}$, $\C{ehf}{no sc}$.
We call ${\cal D}_{\text{\sc basic}}$ the class of all basic graphs
associated to the class (so, if ${\cal D} = \C{Berge}{no cutset}$ then
${\cal D}_{\text{\sc basic}} = \C{Berge}{basic}$ and if ${\cal D} =
\C{ehf}{no sc}$ then ${\cal D}_{\text{\sc basic}} = \C{ehf}{basic}$).

We decompose a graph $G\in {\cal D}$ using extreme 2-joins into basic
graphs.  Let us now define more precisely what we call decomposition
tree (proving its existence and constructing it will be done later).
\vspace{2ex}

\noindent {\em Decomposition tree $T_G$ of depth $p\geq 1$ of a graph
  $G \in {\cal D}$ that has a connected non-path 2-join.}

\begin{enumerate}
\item The root of $T_G$ is $(G^0, {\cal M}^0)$, where $G^0=G$ and
  ${\cal M}^0=\emptyset$.  

\item Each node of the decomposition tree is a pair $(H, {\cal M})$
  where $H$ is a graph of $\cal D$ and ${\cal M}$ is a set of disjoint
  flat paths of length~3 or~4 of $H$.

  The non-leaf nodes of $T_G$ are pairs $(G^0, {\cal M}^0), \ldots,
  (G^{p-1}, {\cal M}^{p-1})$.  Each non-leaf node $(G^i, {\cal M}^i)$
  has two children.  One is $(G^{i+1}, {\cal M}^{i+1})$, the other one
  is $(G^{i+1}_B, {\cal M}^{i+1}_B)$.
 
  The leaf-nodes of $T_G$ are the pairs $(G^1_B, {\cal M}^1_B)$,
  \ldots, $(G^{p}_B, {\cal M}^{p}_B)$ and $(G^{p}, {\cal M}^{p})$.
  Graphs $G^1_B, \ldots ,G^{p}_B$ all belong to ${\cal D}_{\text{\sc
      basic}}$.

\item For $i=0, \ldots, p-1$, $G^i$ has a connected non-path 2-join
  $(X_1^i,X_2^i)$ that is extreme with extreme side $X_1^i$ and that
  is ${\cal M}^i$-independent.   Graphs $G^{i+1}$ and $G^{i+1}_B$ are
  the parity preserving blocks of $G^i$ w.r.t.\ $(X^i_1, X^i_2)$  (as
  defined in Section~\ref{sec:blocks}), whose marker paths are of
  length 3 or~4.  The block $G^{i+1}_B$ corresponds to the extreme
  side $X^i_1$, i.e.\ $X^i_1 \subseteq V(G^{i+1}_B)$. 
 
  Set ${\cal M}^{i+1}_B$ consists of paths from ${\cal M}^i$ 
  whose vertices are in $X^i_1$.  Note that the marker path used to
  construct the block $G^{i+1}_B$ does not belong to ${\cal
    M}^{i+1}_B$.

  Set ${\cal M}^{i+1}$ consists of paths from ${\cal M}^i$ whose
  vertices are in $X^i_2$ together with the marker path $P^{i+1}$
  used to build $G^{i+1}$.

\item ${\cal M}^1_B \cup \ldots \cup {\cal M}^{p}_B \cup {\cal M}^p$
  is the set of all marker paths used in the construction of the nodes
  $G^1, \dots, G^{p}$ of $T_G$, and sets ${\cal M}^1_B , \ldots ,
  {\cal M}^{p}_B, {\cal M}^p$ are pairwise disjoint.
\end{enumerate}

Node $(G^{p}, {\cal M}^{p})$ is a leaf of $T_G$ and is called the
\emph{deepest node} of $T_G$.  Note that all leaves of $T_G$ except
possibly the deepest node are basic.

\begin{lemma}
  \label{l:depth}
  For any decomposition tree $T_G$, the depth of $T_G$ is at most~$n$. 
\end{lemma}

\begin{proof}
  Let a {\em branch} of a graph be any path of length at least~2 whose
  endvertices are both of degree at least 3 and whose interior
  vertices are of degree 2 (in the graph).  For a graph $G$, let $\nu
  (G)$ be the number of vertices of degree at least 3 in $G$, and
  $\tau (G)$ the number of branches in $G$.  We will show that for
  $i=0, \ldots ,p-1$:

  \begin{claim}
    \label{c:ineq}
    $\nu (G^{i+1}) + \tau (G^{i+1})< \nu (G^i) + \tau (G^i)$.
  \end{claim}

  This implies the lemma because it shows that $p$ is at most 
  $$\nu (G)  +\tau (G) \leq n.$$

  Let $i \in \{ 0, \ldots ,p-1\}$ and let $(X_1^i ,X_2^i, A_1^i,
  B_1^i, A_2^i, B_2^i)$ be a split of the 2-join $(X_1^i,X_2^i)$.  Let
  $P^{i+1}=a_1^i \tp \cdots \tp b_1^i$ be the marker path of
  $G^{i+1}$.  Let $P=u \tp \cdots \tp u'$ be a path of $G^{i}$ such
  that $u \in A_1^i,u'\in B_1^i$ and $P \setminus \{ u,u'\} \in X_1^i
  \setminus (A_1^i \cup B_1^i)$ (note that such a path exists since
  $(X_1^i ,X_2^i)$ is connected).  We choose such a path $P$ with a
  minimum number of vertices of degree at least~3.  Since $(X_1^i,
  X_2^i)$ is a non-path 2-join, there is a vertex $q \in X_1^i
  \setminus P$.  Observe that $d_{G^{i+1}} (a_1^i) \leq d_{G^i} (u)$
  and $d_{G^{i+1}} (b_1^i) \leq d_{G^i} (u')$, and hence $\nu
  (G^{i+1})\leq \nu (G^i)$.  When $q$ is of degree at least~3 we have
  $\nu (G^{i+1}) < \nu (G^i)$.  Also, since exactly one branch of
  $G^{i+1}$ intersects $P^{i+1}$, we have $\tau (G^{i+1})\leq \tau
  (G^i) + 1$.

  First suppose that $P$ contains a vertex of degree 2. Then there is
  a branch $P^*$ of $G^i$ that contains a node of $P$, and hence $\tau
  (G^{i+1}) \leq \tau (G^i)$.  If $d_{G^i} (q) \geq 3$, then $\nu
  (G^{i+1}) < \nu (G^i)$, and hence (\ref{c:ineq}) holds. So suppose
  that $d_{G^i} (q) =2$. Then $q$ belongs to a branch $Q^*$ of
  $G^i$. Clearly $P^* \neq Q^*$, so $\tau (G^{i+1}) < \tau (G^i)$, and
  hence (\ref{c:ineq}) holds.

  Now we may assume that all vertices of $P$ are of degree at least 3.
  Suppose that $P$ is of length at least 2. So $\nu (G^{i+1}) \leq \nu
  (G^i)-1$. If $d_{G^i} (q) \geq 3$ then $\nu (G^{i+1}) \leq \nu
  (G^i)-2$, and hence (\ref{c:ineq}) holds. Otherwise, $q$ belongs to
  a branch of $G^i$ and so $\tau (G^{i+1}) \leq \tau (G^i)$, and hence
  (\ref{c:ineq}) holds.

  Finally we assume that $P$ is of length 1, and both $u$ and $u'$ are
  of degree at least 3. By Lemma~\ref{l.starcutset} and
  Lemma~\ref{k1},~\ref{i:BiComp1} and~\ref{i:AiComp1},
  $|A_1^i|,|B_1^i| \geq 2$. Since $A_2^i \cup \{ u,u'\}$ is not a star
  cutset of $G^i$, there is a path $T=t \tp \cdots \tp t'$ such that
  $t \in A_1^i\sm \{u\}, t'\in B_1^i\sm \{u'\}$ and $T \setminus \{
  t,t'\} \subseteq X_1^i \setminus (A_1^i \cup B_1^i)$.  By the choice
  of $P$, $T$ contains at least two vertices of degree at least~3.
  So, $\nu (G^{i+1}) \leq \nu (G^i)-2$ and hence (\ref{c:ineq}) holds.
\end{proof}

\begin{lemma}
  \label{l:buildTree}
  There is an algorithm with the following specification:
  \begin{description}
  \item[Input: ] A graph $G$ in $\cal D$ that has a connected non-path
    2-join.
  \item[Output: ] A decomposition tree $T_G$ of $G$ of depth at most
    $n$ whose leaves are all in ${\cal D}_{\text{\sc basic}}$.
  \item[Running time: ] ${\cal O}(n^4m)$
  \end{description} 
\end{lemma}

\begin{proof}
  Let the root of $T_G$ be $(G^0,{\cal M}^0)=(G,\emptyset)$.  We
  suppose by induction that a decomposition tree of depth $i$ has been
  constructed.  So, the deepest leaf $(G^i,{\cal M}^i)$ is a pair such
  that $G^i \in {\cal D}$ and ${\cal M}^i$ is a set of vertex-disjoint
  flat paths of $G^i$ of length 3 or 4.  Apply the algorithm from
  Lemma~\ref{alg:extreme} to $G^i$ and ${\cal M}^i$.  One of the
  following is the output of this algorithm.

  \vspace{2ex}

  \noindent {\bf Case 1:} An extreme connected non-path 2-join
  $(X_1^i,X_2^i)$ of $G^i$, with say $X_1^i$ being the extreme side,
  that is ${\cal M}^i$-independent.

  Let $G^{i+1}$ and $G^{i+1}_B$ be parity-preserving blocks of
  decomposition of $G^i$ w.r.t.\ $(X_1^i,X_2^i)$ (as defined in
  Section~\ref{sec:blocks}), whose marker paths are of length 3 or 4,
  and block $G^{i+1}_B$ corresponds to $X_1^i$-side.  By
  Lemma~\ref{k:even} and Lemma~\ref{l:recurseBerge}, $G^{i+1}$ and
  $G^{i+1}_B$ are both in ${\cal D}$.  Since $G^{i+1}_B$ is the block
  that corresponds to $X_1^i$, $G^{i+1}_B$ has no connected non-path
  2-join.  If ${\cal D}=\C{ehf}{no sc}$ then by Theorem~\ref{ehf},
  $G^{i+1}_B \in \C{ehf}{basic}$.  If ${\cal D}=\C{Berge}{no cutset}$,
  then by Theorem~\ref{th.3}, $G^{i+1}_B \in \C{Berge}{basic}$.
  Therefore, $G^{i+1}_B \in {\cal D}_{\text{\sc basic}}$.

  Since $(X_1^i, X_2^i)$ is ${\cal M} ^i$-independent, for $P \in
  {\cal M}^i$, either $P \subseteq X_1^i$ or $P \subseteq X_2^i$. Let
  ${\cal M}^{i+1}_B$ be the set of paths of ${\cal M}^i$ that belong
  to $X_1^i$.  Let ${\cal M}^{i+1}$ be the set of paths of ${\cal
    M}^i$ that belong to $X_2^i$ together with the marker path
  $P^{i+1}$ of $G^{i+1}$.  Clearly ${\cal M}^{i+1}_B$ (resp.\ ${\cal
    M}^{i+1}$) is a set of vertex-disjoint flat paths of $G^{i+1}_B$
  (resp.\ $G^{i+1}$) of length~3 or~4, ${\cal M}^{i+1}_B \cap {\cal
    M}^{i+1} =\emptyset$ and ${\cal M}^{i+1}_B \cup {\cal M}^{i+1} =
  {\cal M}^i \cup \{P^{i+1}\}$.

  Hence, in Case~1, we have built a deeper decomposition tree of $G$.
  
  \vspace{2ex}

  \noindent
  {\bf Case 2:} $G^i$ or a block of decomposition of $G^i$ w.r.t.\ some 2-join
  whose marker path is of length at least 3, has a star cutset.

  Since $G^i \in {\cal D}$ and by Lemma \ref{l.starcutset}, $G^i$
  cannot have a star cutset. By Lemma \ref{k1} \ref{i:con}, any 2-join
  of $G^i$ is connected, and hence by Lemma \ref{l:preservesNoSC} this
  case actually cannot happen.

  \vspace{2ex}

  \noindent
  {\bf Case 3:} $G^i$ has no connected non-path 2-join.

  Note that $i\geq 1$ since $G$ is assumed to have a non-path
  connected 2-join.  If ${\cal D}=\C{ehf}{no sc}$ then by Theorem
  \ref{ehf}, $G^i \in \C {ehf}{basic}$. If ${\cal D}=\C{Berge}{no
    cutset}$, then by Theorem~\ref{th.3}, $G^i \in
  \C{Berge}{basic}$. Therefore $G^i \in {\cal D}_{\text{\sc basic}}$.

\vspace{2ex}

By Lemma~\ref{l:depth}, we see that Case~3 must happen at some point,
after at most $n$ iterations.  So, when Case~3 happens we output the
tree $T_G$ and stop.  All the leaves of $T_G$ are basic.  Since the
complexity of the algorithm from Lemma~\ref{alg:extreme} is ${\cal
  O}(n^3m)$ and there are at most $n$ iterations, the algorithm for
constructing $T_G$ runs in time ${\cal O}(n^4m)$.
\end{proof}

\subsection{Clique-decomposition tree}
\label{sec:reprocess}

The \emph{clique-decomposition tree} is used to compute
maximum cliques for graphs in $\cal D$.  This tree $T^C_G$ has the
same definition as $T_G$ except that weights are given to the
vertices.  Let us be more precise by defining how to compute the
children of $(G^i, {\cal M}^i)$ in~$T^C_G$.  Recall that the graph
$G^i$ has an extreme 2-join $(X^i_1, X^i_2)$ with extreme side
$X^i_1$.  Its children are $(G_B^{i+1}, {\cal M}^{i+1}_B)$ and
$(G^{i+1}, {\cal M}^{i+1})$.  In $G^{i+1}_B$, all vertices from
$G[X^i_1]$ keep their weights (as they are in $G^i$) and vertices of
the new marker path receive weight~0.  The weights of $G^{i+1}$ are
assigned as in the construction of the clique block in
Section~\ref{sec:clique}.  Note that computing the weights in the
construction of $G^{i+1}$ relies on several computations on
$G^{i}[X^i_1]$, or equivalently on $G_B^{i+1}$ which is a basic graph.
So by Lemma~\ref{l:optBasic}, building $G^{i+1}$ takes time ${\cal
  O}(n^5)$.  This construction leads to the following result:

\begin{lemma}
  \label{l:Clique}
  There is an algorithm with the following specification:
  \begin{description}
  \item[Input: ] A weighted graph $G$ that is in $\C{Berge}{no
      cutset}$ or in $\C{ehf}{no sc}$ and that has a connected
    non-path 2-join.
  \item[Output: ] A maximum weighted clique of $G$.
  \item[Running time: ] ${\cal O} (n^6)$
  \end{description}
\end{lemma}

\begin{proof}
  By Lemma~\ref{l:buildTree}, we build a decomposition tree for $G$,
  and as explained above we reprocess the tree to get a
  clique-decomposition tree.  By repeatedly applying
  Lemma~\ref{komega}, we see that $\omega(G^0) = \omega(G^1) = \cdots
  = \omega(G^{p})$.  We can compute a maximum weighted clique in the
  basic graph $G^{p}$ by Lemma~\ref{l:optBasic}, and the proof of
  Lemma~\ref{komega} shows how to backtrack such a maximum weighted
  clique to $G$.
\end{proof}

Note that for graphs of $\C{ehf}{no sc}$, the lemma above is not so
interesting because a faster algorithm exists for the class
$\C{ehf}{}$, see the introduction.

\subsection{Stable-decomposition tree for $\C{Berge}{no cutset}$}

We define now the \emph{stable-decomposition tree} $T^S_G$ of a graph
in $\C{Berge}{no cutset}$.  The tree $T^S_G$ has the same definition
as $T_G$ except that even or odd blocks are used sometimes and that
the sets ${\cal M}^i$'s and ${\cal M}^i_B$'s will be sets of disjoint
flat claws and vaults (instead of paths).  Let us be more precise.

\vspace{2ex}

\noindent {\em Decomposition tree $T^S_G$ of depth $p\geq 1$ of a
  weighted graph $G \in \C{Berge}{no cutset}$ that has a connected
  non-path 2-join.}

\begin{enumerate}
\item The root of $T^S_G$ is $({G''}^0, {{\cal M}''}^0)$, where ${G''}^0=G$ and
  ${\cal M}''^0=\emptyset$.

\item For each node $(H, {\cal M})$ of the decomposition tree, $H$ is
  a Berge graph and ${\cal {\cal M}}$ is a set of disjoint flat claws
  or flat vaults of $H$.

  The non-leaf nodes of $T^S_G$ are pairs $(G''^0, {\cal M}''^0),
  \ldots, (G''^{p-1}, {\cal M}''^{p-1})$.  Each non-leaf node $(G''^i,
  {\cal M}''^i)$ has two children.  One is $(G''^{i+1}, {\cal
    M}''^{i+1})$, the other one is $(G''^{i+1}_B, {\cal
    M}''^{i+1}_B)$.
 
  The leaf-nodes of $T^S_G$ are the pairs $(G''^1_B, {\cal M}''^1_B)$,
  \ldots, $(G''^{p}_B, {\cal M}''^{p}_B)$ and $(G''^{p}, {\cal
    M}''^{p})$.  Graphs $G''^1_B, \ldots ,G''^{p}_B$ are all
  parity-preserving extensions of graphs from $\C{Berge}{basic}$.

\item For $i=0, \ldots, p-1$, $G''^i$ has a connected non-path 2-join
  $(X_1''^i, X_2''^i)$.

  Graph $G''^{i+1}_B$ is the parity-preserving block of $G^i$ w.r.t.\
  $(X''^i_1, X''^i_2)$ (as defined in Section~\ref{sec:blocks}), whose
  marker path is of length 3 or~4 and which corresponds to the side
  $X''^i_1$, i.e.\ $X''^i_1 \subseteq V(G''^{i+1}_B)$.  Vertices of
  $X''^i_1$ keep their weight from $G''^i$ and vertices of the marker
  path receive weight zero.
 
  Graph $G''^{i+1}$ is the even block or the odd block of $G''^i$ w.r.t.\
  $(X''^i_1, X''^i_2)$, according to the $X''_1$-parity of $(X''^i_1,
  X''^i_2)$. 

  Set ${\cal M}''^{i+1}_B$ consists of claws and vaults from ${\cal
    M}''^i$ whose vertices are in $X''^i_1$.  Note that the marker path
  used to construct the block $G''^{i+1}_B$ does not belong to ${\cal
    M}''^{i+1}_B$.

  Set ${\cal M}''^{i+1}$ consists of claws and vaults from ${\cal M}''^i$
  whose vertices are in $X''^i_2$ together with the claw or the vault
  $P''^{i+1}$ used to build $G''^{i+1}$.

\item ${\cal M}''^1_B \cup \ldots \cup {\cal M}''^{p}_B \cup {\cal M}''^p$
  is the set of all marker claws or vaults used in the construction of
  the nodes $G''^1, \dots, G''^{p}$ of $T^S_G$, and sets ${\cal M}''^1_B ,
  \ldots , {\cal M}''^{p}_B, {\cal M}''^p$ are pairwise disjoint.
\end{enumerate}

The existence of $T^S_G$ is not clear since introducing even and odd
blocks may create star cutsets (and so balanced skew partitions) in
our graphs, so that we cannot rely on Theorem~\ref{th.3} to build the
tree recursively.  But here we show how to actually construct $T^S_G$
by reprocessing $T_G$.

Start from $T_G$ and for each node $(G^i, {\cal M}^i)$, $i=1, \dots,
p$ of $T_G$, extend the marker path introduced in that node to obtain
a graph ${G''}^i$.  Accordingly, replace the marker paths in the
graphs $G^i_B$ and sets ${\cal M}^i_B$, ${\cal M}^i$ by marker claws
and vaults to obtain $G''^i_B$, ${\cal M}''^i_B$ and ${\cal M}''^i$.
We obtain the nodes $({G''}^{i}, {\cal M}''^{i})$ and $({G''}^{i}_B,
{\cal M}''^{i}_B)$ of $T^S_G$.

Note that extending a flat path $P$ of length at least 3 in a graph
$H$ with a $\{P\}$-independent connected 2-join $(X_1, X_2)$, yields a
graph $H''$ that has a connected 2-join $(X''_1, X''_2)$ naturally
arising from $(X_1, X_2)$: put all vertices of $X_1 \sm P$ in $X''_1$,
all vertices of $X_2\sm P$ in $X''_2$ and put the claw or the vault
arising from $P$ in $X''_1$ when $P\subseteq X_1$ and in $X''_2$ when
$P\subseteq X_2$.  So, the connected 2-joins $({X''}^i_1,
{X''}^i_2)$'s all exist and are immediate to find from $T_G$.

Note that by Lemma~\ref{l:stayBerge}, all nodes of $T^S_G$ are in
$\C{Berge}{}$.  So, all the 2-joins $({X''}^i_1, {X''}^i_2)$'s are
either ${X''}^i_1$-even or ${X''}^i_1$-odd and we can choose whether
we use an even or an odd block.

Note that computing the weights in the construction of the even or odd
block ${G''}^{i+1}$ relies on several computations on
${G''}^{i}[X^i_1]$, or equivalently on ${G''}_B^{i+1}$ which is a
parity-preserving extension of a graph from $\C{Berge}{basic}$, so the
computations can be done in time $O(n^5)$ by Lemma~\ref{l:optBasic}.

\subsection{Stable-decomposition tree for $\C{ehf}{no sc}$}

For an even-hole-free graph the tree $T^S_G$ as defined above may fail
to exist.  Because building an odd block does not preserve being
even-hole-free.  So, 2-joins appearing in the decomposition tree may
fail to be either $X_1$-even or $X_1$-odd, so we do not know when to
use even or odd blocks. But with a little twist, we can define a
useful tree.  

The definition of $T^S_G$ for a graph in $\C{ehf}{no sc}$ is very
similar to the definition for $\C{Berge}{no cutset}$, so we do not
repeat it and point out the differences instead.

The main difference is when to use odd or even block.  Recall that
$G''^i$ has a 2-join $(X''^i_1, X''^i_2)$.  Let $(X''^i_1, X''^i_2, A''^i_1,
B''^i_1, A''^i_2, B''^i_2)$ be a split of this 2-join.  We define $a^i =
\alpha(G[{A''^i_1 \cup C''^i_1}])$, $b^i= \alpha(G[{B''^i_1 \cup C''^i_1}])$,
$c^i = \alpha(G[{C''^i_1}])$ and $d^i = \alpha(G[{X''^i_1}])$.  If $a^i +
b^i \leq c^i + d^i$ then $G''^{i+1}$ is the even block $G''^i$ w.r.t.\
$(X''^i_1, X''^i_2)$.  Else, it is the odd block.  Note that graphs in
$T_G^S$ are not required to be in $\C{ehf}{}$. 

The other difference is that graphs $G''^1_B, \ldots ,G''^{p}_B$ and
$G''^p$ are any (and not only parity-preserving) extensions of graphs
from $\C{ehf}{basic}$.  Note that in Lemma~\ref{l:optBasic}, the
extensions are not required to be parity-preserving for
$\C{ehf}{basic}$.

\begin{lemma}
  \label{l:Stable}
  There is an algorithm with the following specification:
  \begin{description}
  \item[Input: ] A weighted graph $G$ that is in $\C{Berge}{no
      cutset}$ or in $\C{ehf}{no sc}$ and that has a connected
    non-path 2-join.
  \item[Output: ] A maximum weighted stable set of $G$.
  \item[Running time: ] ${\cal O} (n^4m)$
  \end{description}
\end{lemma}

\begin{proof}
  By Lemma~\ref{l:buildTree}, we build a decomposition tree for $G$,
  and as explained above we reprocess the tree to get a
  stable-decomposition tree $T^S_G$.  By repeatedly applying
  Lemmas~\ref{l:evenBlock} and~\ref{l:oddBlock}, we see that
  $\alpha({G''}^0) = \alpha({G''}^1) = \cdots = \alpha({G''}^{p})$.
  Note that the inequalities necessary to apply
  Lemmas~\ref{l:evenBlock} and~\ref{l:oddBlock} are satisfied.  For
  Berge graphs, this follows from the fact that all nodes of $T^S_G$
  are Berge so that we can rely on Lemma~\ref{l:ineqEven} and
  \ref{l:ineqOdd}.  For even-hole-free graphs, the inequalities are
  true from the definition of $T^S_G$.

  We can compute a maximum weighted stable set in the extension of
  basic graph ${G''}^{p}$ by Lemma~\ref{l:optBasic} (when $G$ is
  Berge, the extension is parity-preserving) , and the proofs of
  Lemmas~\ref{l:evenBlock} and~\ref{l:oddBlock} show how to backtrack
  such a maximum weighted stable set to $G$.
\end{proof}

\section{Optimization algorithms}
\label{sec:color}

\begin{theorem}
  \label{l:CS}
  There is an algorithm with the following specification:
  \begin{description}
  \item[Input: ] A weighted graph $G$ that is either a Berge graph
    with no balanced skew partition, no connected non-path 2-join in
    the complement and no homogeneous pair; or an even-hole-free graph
    with no star cutset.
  \item[Output: ] A maximum weighted stable set and a maximum weighted
    clique of $G$.
  \item[Running time: ] ${\cal O} (n^6)$
  \end{description}
\end{theorem}

\begin{proof}
  If $G$ is in $\C{Berge}{basic}$ or in $\C{ehf}{basic}$ we rely on
  Lemma~\ref{l:optBasic}.  Else, by
  Theorem~\ref{ehf} or~\ref{th.3}, $G$ has a connected non-path 2-join.
  So, we may rely on Lemmas~\ref{l:Clique} and~\ref{l:Stable}.
\end{proof}

Let us now point out that our method works for our two classes for
different reasons.  For Berge graphs, it is because our even and odd
blocks are class-preserving. For even-hole-free graphs, it is because
the basic class is restricted to graphs obtained from line graphs by
adding a bounded number of vertices.  In fact, our method works for
something more general than even-hole-free graphs with no star cutset.
Let $L_k$ be the class of graphs obtained from line graphs by adding
$k$ vertices.  If $k$ is fixed, we can compute maximum weighted cliques
and stable sets for any class that is decomposable with extreme
2-joins into graphs of $L_k$. For $k=0$, this gives a subclass of
claw-free graphs. For $k=2$ this gives a super-class of even-hole-free
graphs with no star cutset.

\begin{theorem}
  \label{th:color}
  There exists an algorithm of complexity ${\cal O}(n^{7})$ whose
  input is a Berge graph $G$ with no balanced skew partition, no
  connected non-path 2-join in the complement and no homogeneous pair,
  and whose output is an optimal coloring of $G$ and an optimal
  coloring of $\overline{G}$.
\end{theorem}

\begin{proof}
  There exists a combinatorial coloring algorithm for an input perfect
  graph $G$ that uses at most $n$ times as subroutines algorithms
  for maximum cliques and stable sets.  See~\cite{KrSe:colorP} or
  Corollary 67.2c in \cite{schrijver}.  This algorithm relies on the
  fact that perfect graphs are closed under taking induced subgraphs
  and replicating vertices.  Our class is not, but taking induced
  subgraphs is easily simulated by giving weight~0 to a vertex and
  replicating $k$ times a vertex is simulated by giving weight $k$ to
  the vertex.  The method also works for $\overline{G}$ because we may
  compute maximum weighted cliques and stable sets for $\overline{G}$
  as well. 
\end{proof}

\section{NP-completeness}
\label{sec:npc}

Here, we give a class $\cal C$ of graph for which computing a maximum
stable set is NP-hard.  The interesting feature of class $\cal C$ is
that all graphs in $\cal C$ are decomposable along extreme 2-joins
into one bipartite graph and several gem-wheels where a
\emph{gem-wheel} is any graph made of an induced cycle of length at
least 5 together with a vertex adjacent to exactly four consecutive
vertices of the cycle.  Note that a gem-wheel is a line graph (of a
cycle with one chord).  Our NP-completeness result (proved jointly
with Guyslain Naves) shows that being able to decompose along extreme
2-joins is not enough in general to compute stables sets. This
suggests that being in $\C{parity}{}$ is essential for computing
stable sets along 2-joins and that the inequalities of
Subsection~\ref{sec:ineq} capture an essential feature of
$\C{parity}{}$.

Here, \emph{extending} a flat path $P = p_1\tp \cdots \tp p_k$ of a
graph means deleting the interior vertices of $P$ and adding three
vertices $x, y, z$ and the following edges: $p_1x$, $xy$, $yp_k$,
$zp_1$, $zx$, $zy$, $zp_k$.  By extending a graph $G$ we mean
extending all paths of $\cal M$ where $\cal M$ is a set a flat paths
of length at least 3 of $G$.  Class $\cal C$ is the class of all
graphs obtained by extending 2-connected bipartite graphs.  From the
definition, it is clear that all graphs of $\cal C$ are decomposable
along extreme connected non-path 2-joins.  One leaf of the
decomposition tree will be the underlying bipartite graph.  All the
others leaves will be gem-wheels.

We call \emph{4-subdivision} any graph $G$ obtained from a graph $H$
by subdividing four times every edge.  More precisely, every edge $uv$
of $H$ is replaced by an induced path $u\tp a \tp b\tp c\tp d \tp v$
where $a, b, c, d$ are of degree two.  It is easy to see that
$\alpha(G) = \alpha(H) + 2|E(H)|$. This construction, essentially due
to Poljak~\cite{poljak:74}, yields as observed by Guyslain Naves:

\begin{theorem}[Naves, \cite{naves:pc}]
  \label{th:npHard}
  The problem whose instance is a graph $G$ from $\cal C$ and an
  integer $k$, and whose question is ``Does $G$ contain a stable set
  of size at least~$k$'' is NP-complete.
\end{theorem}

\begin{proof}
  Let $H$ be any graph.  First we subdivide 5 times every edge of~$H$.
  So each edge $ab$ is replaced by $P_7 = a \tp p_1 \tp \cdots \tp p_5
  \tp b$.  The graph $H'$ obtained is bipartite.  Now we build an
  extension $G$ of $H'$ by replacing all the $P_5$'s $p_1 \tp \cdots
  \tp p_5$ arising from the subdivisions in the previous step by
  $P_4$'s.  And for each $P_4$ we add a new vertex complete to it and
  we call \emph{apex vertices} all these new vertices.  The graph $G$
  that we obtain is in $\cal C$.  It is easy to see that there exists
  a maximum stable set of $G$ that contain no apex vertex because an
  apex vertex of a maximum stable set can be replaced by one vertex
  of its neighborhood.  So, we call $G'$ the graph obtained from $G$
  by deleting all the apex vertices and see that $\alpha(G') =
  \alpha(G)$.  Also, $G'$ is the 4-subdivision arising from $H$.  So
  from the remark above, maximum stable sets in $H$ and $G$ have sizes
  that differ by $2|E(H)|$.
\end{proof}

Our NP-completeness result is related to Lemma~\ref{lineblock}.  With
respect to computing maximum stable sets, these two results say in a
sense that gem-blocks carry enough information to encode one side of a
2-join.

\section*{Acknowledgement}

This work started in November 2008 when the first author was visiting
Universidade Federal do Rio de Janeiro, Brasil, during discussions
with Simone Dantas, Sulamita Klein and Celina de Figueiredo.
Originally, the project was to color Berge graphs without balanced
skew partitions.  Celina suggested to add the homogeneous pair-free
assumption and this idea was very important for the success of this
work.

Thanks to Guyslain Naves for showing the NP-hardness result.   This
manuscript also benefited from a very careful reading of an anonymous
referee.

\end{document}